\documentclass{article}
\usepackage{fullpage}
\usepackage{xspace}
\usepackage{indentfirst}
\usepackage{amsmath}

\newtheorem{remark}{Remark}[section]
\usepackage[breaklinks=true]{hyperref}

\usepackage[capitalise,noabbrev,nameinlink]{cleveref}
\usepackage{adjustbox}
\usepackage{todonotes}
\usepackage{arydshln}
\usepackage{fancybox}
\usepackage{algorithm,algpseudocode}
\AtBeginDocument{%
  }

\newtheorem{definition}{Definition}
\newtheorem{theorem}{Theorem}
\newtheorem{proof}{Proof}
\usepackage{amsfonts} 

\usepackage{booktabs}
\usepackage{newpxtext,newpxmath}
\newcommand{\NGFPSU}{NullG-FPSU\xspace}
\newcommand{\NFFPSU}{1DNF-FPSU\xspace}
\newcommand{\LAYFPSU}{LAYER-FPSU\xspace}
\newcommand{\EXCFPSU}{EXCLS-FPSU\xspace}
\newcommand{\STRFPSU}{STRIP-FPSU\xspace}

\newcommand{\SETUP}{Setup\xspace}  
\newcommand{\QUERY}{Query\xspace}  
\newcommand{\RESP}{Respond\xspace} 
\newcommand{\DECOD}{Decode\xspace} 
\newcommand{\ENCOD}{Encode\xspace} 
\newcommand{\UNION}{Union\xspace}  
\usepackage[affil-it]{authblk}
\begin{document}

\title{Fuzzy Private Set Union via Oblivious Key Homomorphic Encryption Retrieval}

\author{Jean-Guillaume Dumas \thanks{Univ. Grenoble Alpes. Laboratoire Jean Kuntzmann, CNRS, UMR 5224. 150 place du Torrent, IMAG -
CS 40700, 38058 Grenoble, cedex 9 France. {jean-guillaume.dumas}@univ-grenoble-alpes.fr}, Aude Maignan
  \thanks{Univ. Grenoble Alpes. Laboratoire Jean Kuntzmann, CNRS, UMR 5224. 150 place du Torrent, IMAG -
CS 40700, 38058 Grenoble, cedex 9 France. {aude.maignan}@univ-grenoble-alpes.fr}, Luiza Soezima \thanks{Aarhus University, Aarhus, Denmark, reisbs.luiza@gmail.com}}

\maketitle
\begin{abstract}
Private Set Multi-Party Computations are protocols that allow parties
to jointly and securely compute functions: apart from what is
deducible from the output of the function, the input sets are kept
private.
Then, a Private Set Union (PSU), resp. Intersection (PSI), is a
protocol that allows parties to jointly compute the union, resp. the
intersection, between their private sets.
Now a structured PSI, is a PSI where some structure of
the sets can allow for more efficient protocols.
For instance in Fuzzy PSI, elements only need to be close enough,
instead of equal, to be part of the intersection.
We present in this paper, Fuzzy PSU protocols (FPSU), able to efficiently
take into account approximations in the union.
For this, we introduce a new efficient sub-protocol,
called Oblivious Key Homomorphic Encryption Retrieval (OKHER),
improving on Oblivious Key-Value Retrieval (OKVR) techniques in our
setting.

In the fuzzy context, the receiver set $X=\{x_i\}_{1..n}$ is replaced
by ${\mathcal B}_\delta (X)$, the union of $n$ balls of dimension $d$
with radius $\delta$, centered at the $x_i$. The sender set is just
its $m$ points of dimension $d$. Then the FPSU functionality
corresponds to $X \sqcup \{y \in Y, y \notin {\mathcal B}_\delta (X)\}$.
Thus, we formally define the FPSU functionality and security
properties, and propose several protocols tuned to the patterns of the
balls using the $l_\infty$ distance.
Using our OKHER routine and homomorphic encryption,
we are for instance able to obtain a FPSU protocols with an asymptotic
communication volume bound ranging from $O(dm\log(\delta{n}))$ to
$O(d^2m\log(\delta^2n))$, depending on the receiver data set structure.
\end{abstract}

\section{Introduction}
In the literature of secure Multi-Party Computations, there is a particular field of study that focuses on
securely computing set operations \cite{BS05,KS05,DBLP:conf/acns/Frikken07,DBLP:conf/acisp/DavidsonC17},
such as private intersection, union, or difference, between two
parties, a client and a server (or a sender and a receiver).
Designing privacy-preserving protocols realizing such functionalities
include dealing with trade-offs
between communication volume and rounds, as well as with computational
costs for both the sender and receiver.
While traditional protocols rely on exact matching between set
elements, there has been a recent interest on the relaxed notion of
comparisons, a \textit{fuzzy matching}, and more generally on taking
advantage of some structure of the sets in order to allow for  more efficient
protocols~\cite{Chmielewski:2008:Fuzzy,10.1007/978-3-031-15802-5_12,10.1007/978-3-031-15802-5_12,cryptoeprint:2025/054}.

For instance a Private Set Union (PSU) is a protocol used in scenarios
where information are combined from different sources while
maintaining input privacy. In PSU, at least one party learns the
union of the sets while no party should learn the intersection.
Extending PSU protocols to support fuzzy matching opens a new
set of studies in privacy-preserving computation.
In many applications, elements could be considered equal if they in fact
fall within a defined threshold under some distance metric.
While this notion of fuzziness has been studied in the context of
Private Set Intersection (PSI), its adaptation to PSU remains
challenging,
due to the significantly higher computational and communication costs
induced by distance-based matching over large domains.
These challenges arise in practical applications such as biometric
databases, record linkage, and data integration across heterogeneous
sources, where identifiers are noisy and approximate matching is
required, yet scalability and strict privacy guarantees must be preserved.

Despite this clear practical demand, to the best of our knowledge
there is no protocol in the existing literature that formally
addresses \emph{Fuzzy Private Set Union} (FPSU) as a standalone primitive.
Prior works on fuzzy or approximate matching focus almost exclusively
on intersection-based functionalities and do not directly extend to
union semantics without incurring excessive leakage or prohibitive
overhead.
Some works on (unbalanced) Private Set Union could apply to the fuzzy
setting, but they usually need to consider the whole space of nearby
elements, thus incurring huge costs, at least for the receiver, as
shown in~\cref{tab:new}.
This paper fills this gap by introducing a dedicated
study of Fuzzy PSU, formalizing its functionality and security
properties, identifying its inherent challenges, and proposing
efficient protocols that enables privacy-preserving unions over
approximate equivalence classes of structured sets.

\subsection{Related works}

\subsubsection{Fuzzy Private Set Intersection}
A Private Set Intersection (PSI) protocol allows parties to compute
the intersection of their sets, without revealing any other
information.
Many protocols have been
proposed to achieve this functionality under different applications (such as private contact discovery)
or security models (semi-honest or malicious).
Recent advances, use techniques such as \textit{Oblivious Pseudorandom Function (OPRF)} \cite{10.1145/2976749.2978381, 10.1007/978-3-030-56877-1_2, 10.1007/978-3-030-26954-8_13, Rindal:2021:VOLEPSI, PAXOS},
where the OPRF is usually build from Oblivious Transfers (OT), and
obtain a balance between communications and computations.
Works such as \cite{PAXOS, DBLP:conf/ccs/RaghuramanR22,
  10.1007/978-3-030-84245-1_14} uses \textit{Oblivious Key Value Store
  (OKVS)} to achieve better communication costs.
%
A variant of PSI is the
\emph{Fuzzy Private Set Intersection} (FPSI), in which the parties
consider a fuzzy matching: points below a distance threshold are
considered equal.
Formally, the receiver holds a structured set,
such as the union of $\delta$-radius balls.
The computational and communication costs should then be able to take
into account this structure of the receiver’s set, i.e. the union of
the volume of all the balls  around all the receiver’s elements.

Results in FPSI include~\cite{10.1007/978-3-031-15802-5_12,
  10.1007/978-3-031-68397-8_10}, where the core protocol uses
\textit{boolean Function Secret Sharing (bFSS)}, and their idea relies
on gathering the union of $n$ balls of radius $\delta$ with respect to
the $L_\infty$ norm in $d$-dimensional space.
Here, the dimensions are divided in grids.
Then, after realizing the $\delta$-balls around each item, the balls
will map into intersected grid cells. In the original construction,
correctness is ensured by requiring the balls to be pairwise disjoint.
The use of a grid-based partitioning relaxes this requirement, since
disjointedness no longer needs to
be enforced along each coordinate axis.
However, the costs associated with this idea makes the receiver
computation proportional to the total volume of the input
balls.
Recent work such as \cite{10.1007/978-3-031-58740-5_12} overcomes the
limitations of this previous approach.
The idea is that each block is unique for each disjoint ball, i.e.,
two disjoint balls must be associated with different blocks, in a way
that the receiver encodes the identifier of each block and the sender
would decode it by checking all potential blocks.
Specifically it encodes a key-value store, with the keys
being the projection of the balls together with the threshold
distance, and the values are a secret share associated with a
dimension and an index, together with a hash of the projection of the
ball. With this, they manage to achieve a communication and
computation cost of $\mathcal{O}((\delta d)^2n+m)$, again to be
compared with our~\cref{tab:new}.

\subsubsection{Private Set Union}
A private set union (PSU) protocol is a cryptographic protocol
involving two parties. The receiver, denoted $\mathcal{R}$, owns a set
$\textbf{X}$, and the sender, denoted $\mathcal{S}$, owns a set
\textbf{Y}. The desired functionality of such a protocol
is that the receiver
$\mathcal{R}$ receives only the union $\textbf{X}\cup \textbf{Y}$, and
the sender $\mathcal{S}$ learns nothing.
(Note that this is equivalent to the receiver learning the set difference
$\textbf{Y} {\setminus} \textbf{X}$.)
The
protocol is parameterized with (upper bounds on) the set sizes
$|\textbf{X}|$ and $|\textbf{Y}|$, which are therefore implicitly
revealed to both parties as well. However, the sender $\mathcal{S}$
learns nothing about the content of \textbf{X} and the receiver
$\mathcal{R}$ learns nothing about
$\textbf{X}\cap\textbf{Y}$.
PSU protocols have been widely studied in the case of {\em balanced} input set size~\cite{BS05,KS05,DBLP:conf/acns/Frikken07,DBLP:conf/acisp/DavidsonC17,DBLP:conf/asiacrypt/KolesnikovRT019,DBLP:conf/pkc/GarimellaMRSS21,Jia+22,Zhang:2023:Usenix:LPSU}, motivated by numerous practical applications such as
disease data collection from hospitals.
Finally, there is some recent work in an {\em unbalanced} setting (UPSU), in particular
where the sender's input set is (quite) smaller than the
receiver's~\cite{Tu:2023:CCS:UPSU,10.1145/3658644.3690308,Dumas:2025:upsu}.
This could apply to the fuzzy
setting, but usually one needs to use as points the whole space of nearby
elements, thus incurring huge costs, as
shown in~\cref{tab:new}.

\subsection{Contributions}

{\small \begin{table*}[h]
\caption{Complexity bounds of FPSU protocols, where the receiver $\mathcal{R}$ has a set of  $n$ $d$-dimensional balls of radius $\delta$, $\chi$ is the coloring number of the induced graph $G_X$. And the sender has a set of $m$ $d$-dimensional points with $n>{m}$.}\label{tab:new}
\centering
\begin{tabular}{|c||c|c|c|c|c|}
\hline
\textbf{Protocol}&Graph &Cost for $\mathcal{R}$& Cost for $\mathcal{S}$& Comm. Vol. & $\#$ rounds \\
\hline
\cite[$\mathsf{PSU}_{\mathsf{pk}}$]{10.1145/3658644.3690308}&any&$O(\delta^dn)$&{$O(dm\log (\delta^dn))$}&{$O(dm\log (\delta^dn))$}&$ \geq 3$\\
\cite[$\mathsf{PSU}_{3}$]{Dumas:2025:upsu} &any&$(\delta^dn)^{1+o(1)}$&{$(dm)^{1+o(1)}$}&{$O(dm)$}&3\\
\hline

        \NGFPSU & null&$\mathcal{O}(d\delta n \log(\delta n))$ & $\mathcal{O}(dm \log (\delta n))$ & $\mathcal{O}(dm \log (\delta n))$&3 \\
  \NFFPSU & 1 non-fuz dim.& $\mathcal{O}(d \delta n \log(\delta n))$ & $\mathcal{O}(dm \log (\delta n))$ & $\mathcal{O}(dm \log (\delta n))$&3 \\
\LAYFPSU & any&$\mathcal{O}(d\delta n \log(\delta n)+n^2)$ & $\mathcal{O}(\chi dm \log (\delta n))$ & $\mathcal{O}(\chi dm \log (\delta n))$&4 \\
  \EXCFPSU & exclusive& $\mathcal{O}(d\delta^2  n \log (\delta^2 n))$
  & $\mathcal{O}(d^2m \log (\delta^2 n))$ & $\mathcal{O}(d^2m \log
  (\delta^2 n))$&3 \\
 \STRFPSU &d-stripable& $\mathcal{O}(d\delta^2 n\log (\delta^2
 n)+dn^2)$ & $\mathcal{O}(d^2m \log (\delta^2 n))$ & $\mathcal{O}(d^2m
 \log (\delta^2 n))$&3 \\
\hline
    \end{tabular}
\end{table*}}
Our first contribution is a functionality named Oblivious Key Homomorphic Encryption Retrieval (\textbf{OKHER}). There the server holds a  set of key-value pairs $KV=\{  (k_i,v_i)_{i\in[N]}  \} \subset \mathcal K \times \mathcal C$ where $N$ is large and the clients holds a set of queries $Q=\{(q_i)_{i\in[m]} \} \subset \mathcal K$ such that $m<<N$.
At the end of the \textbf{OKHER} protocol's execution, on the one hand, the
server will have learned nothing about the client's requests. On the
other hand, for all of its requests $q_i$, the client will receive an
encrypted value, equal to the encryption of $v_j$, if there exists a
$j$ such that $q_i=k_j$. Otherwise, it will receive a dummy value
indistinguishable from a random encrypted value $\mathcal C$.
This functionality is a variant of the Oblivious Key-Value Retrieval
(OKVR) one presented in \cite{298226}; but here we do not need a
preliminary oblivious pseudo-randomization (OPRF).
In fact this OPRF ensures that the receiver values are
indistinguishable from a uniformly distribution. We instead require
that our encryption is IND-CPA, with the same effect on the security.

\textbf{OKHER} will be used as a main component of a new \emph{Fuzzy
  Private Set Union} (FPSU) protocol and some variants.
For all of them the receiver owns a set of
d-dimensional points $X= \{x_1, \ldots,x_n\}$ and a fuzzy threshold
$\delta$; and the sender owns a smaller set of  d-dimensional points
$Y= \{y_1, \ldots,y_m\}$.
A first approach could be to apply an UPSU protocol to the $(2 \delta+1)^dn$
elements of the $n$ balls centered in the $x_i$ ($i \in[n]$).
Complexity bounds for this approach are given in the first rows
of~\cref{tab:new}. To avoid the  $\delta^d$ factor, we will consider
projections onto the $d$ different axes (as was already done, e.g.,
for FPSI, in~\cite{10.1007/978-3-031-58740-5_12}).
On top of this, we further construct the induced graph $G_X$ of a set $X$
equipped with a fuzzy threshold $\delta$: the vertices are the centers
of the balls, and they are linked by an edge for each dimensions if
the projection of their balls do intersect.
We then use the different
properties of such graphs to develop more efficent approaches.
In particular, we consider:
\begin{itemize}
\item  \textbf{\NGFPSU}: (FPSU for Null Graphs) This protocol
  (see~\cref{proto:NGFPSU}) can be used when the projections of the balls
  centered at $x_i$ are disjoint along all axes (or more precisely
  when the induced graph is null).

\item  \textbf{\NFFPSU}: (FPSU for Non fully Fuzzy dataset)  This
  protocol 
 is devoted to an application on
  biometric databases where $x_i=(x_{i1}, \ldots, x_{id})$ is composed
  of a non-fuzzy and unique ID $x_{i1}$ and fuzzy biometric scan
  values.
\item  \textbf{\LAYFPSU}: (FPSU by LAYers)  This protocol
(see Subsection \ref{LAYFPSU}) can be used in all cases. It uses a graph
  coloring algorithm on the induced graph $G_X$, but the chromatic
  number (or only an upper bound on it) must be revealed to the sender.

\item \textbf{\EXCFPSU}:  (FPSU for EXClusive graphs) This protocol
(see Subsection \ref{ExcSection})
 can be used when the induced graph is
  a-exclusive (see Definition \ref{defExc}). In particular, it can be
  applied when the projections of the balls centered at $x_i$ are
  disjoint along a single axis.
\item  \textbf{\STRFPSU}: (FPSU for d-STRipable graphs). This protocol
(see Subsection \ref{sec:STRFPSU})
 can be used when the induced graph is
  d-stripable (see Definition \ref{defStrip}). In particular, it can
  be applied when one projection of each ball has no intersection with
  the others.
\end{itemize}
The asymptotic complexity bounds of these five approaches are presented
in~\cref{tab:new}.

\section{Balls and graphs}
In this section, we present some definitions and concepts from graph theory that are useful in our case study.
%
In the context of fuzzy private operations, fuzziness is interpreted by considering balls around the points of $X$.
A ball of radius $\delta$ centered in a point $x_i \in {\mathcal K}^d$ is denoted
${\mathcal B}_\delta (x_i)= [x_{i1}- \delta,x_{i1}+ \delta]
\times\ldots\times [x_{id}- \delta,x_{id}+ \delta]$.  There we denote by
$x_{ij} \in \mathcal{K}$ the $jth$ coordinate of the element $x_i \in
\mathcal{K}^d$.

The fuzzy intersection with respect to $\delta$ becomes
$X \cap_F Y= \left (\bigcup_{j=1}^n {\mathcal B}_\delta (x_j) \right ) \cap Y$
and
$$X \cup_F Y=X \bigsqcup \left\{ y_i \in Y,~\text{s.t.}~y_i \notin \bigcup_{j=1}^n {\mathcal B}_\delta (x_j) \right\} $$

When the balls are not axis-disjoint
(meaning that $\exists l \in[d], \exists (i,j)$    $ \in [n]^2, [x_{il}-
\delta, x_{il}+ \delta] \cap [x_{jl}- \delta, x_{jl}+ \delta] \neq
\emptyset$)
the construction of the key-value pairs are more complex.
Since two balls can then indeed share some projected values, even if
all the keys in the set of key-value pairs must be different.
We propose a point of view using graphs to ease the construction of
the different cases.
\subsection{Mapping fuzziness to graphs}
\begin{definition}

Let $X=\{x_1,x_2, \ldots, x_n\} \subset {\mathcal K}^d$, and $\delta
\in \mathbb N$ . An induced edge-labeled graph $G_X=(V_X,E_X, \phi)$
is defined by its set of vertices $V_X=\{1,2,\ldots, n\}$ and its set
of edges $E_x$ associated with a label map  $\phi: E_X \rightarrow \{1,.., d \}^*$ such that for $i \neq j, (i,j) \in E_X$ and $l \in \phi((i,j))$ if and only if $ l \in [d] \wedge  [x_{il}- \delta,x_{il}+\delta] \cap [x_{jl}-\delta,x_{jl}+ \delta] \neq \emptyset$.

If $(i,j) \notin E_G, \phi((i,j))= \emptyset$ and
for all $i \in V_G$,

$deg(i) =\sum_{j \in V {\setminus} \{i\}} card (\phi((i,j)))$
\end{definition}
\begin{remark}
Remark: Here we consider antireflexive undirected graphs so that $(i,j)=(j,i)$ is only considered once.
\end{remark}
By extension, we also have that: $ deg(G_X)= max_{i \in V_X} deg(i)$.
\begin{remark}
If we consider $G_X$ as a simple graph (by removing $\phi$ to obtain an unweighted, undirected graph containing no graph loops nor multiple edges) then its degree is not greater than $deg(G_X)$ because the multiple edges are not taken into account in the computation of the degree.
\end{remark}
Next, the idea will be that the edges will be labeled and denoted by
their associated dimension, hence the following definition:
\begin{definition}
If $(i,j) \in E$ and $a \in \phi((i,j))$ we say that $(i,j)$ is an
a-edge of $G_X$ and we denote $E^a_X$ the subset of $E_X$ restricted
to a-edges and $G_X^a=(V_X,E^a_X)$ the corresponding spanning subgraph of $G_X$. By extension, $E_X^{a,b}$ is the subset of edges $(i,j)\in E_X$ such that $\{a,b\} \subset \phi((i,j))$.
\end{definition}

When the induced edge-labeled graph is a null graph, the balls are axis-disjoint meaning that the projections of all the balls in any axis  are disjoints.
The opposite case occurs if $(i,j) \in E_X$ and $\phi((i,j))=[1..d]$:
then $ {\mathcal B}_\delta (x_i)\cap {\mathcal B}_\delta (x_j) \neq \emptyset$. In other words, the degree of an edge is at most $d-1$ if the balls are disjoints.
\begin{definition}
A subgraph $G_Z$ of a graph $G_X=(V_X,E_X,\phi)$ is defined by its set of vertices $V_{Z {\setminus} X}$, its set of edges $E_{Z {\setminus} X}= \{(i,j) \in E_X s.t. i \in V_{Z {\setminus} X} \wedge j \in V_{Z {\setminus} X} \}$ and the label map $\phi' : E_{Z {\setminus} X} \rightarrow \{1,.., d \}^*$ such that if $(i,j) \in E_{Z {\setminus} X}, \phi'((i,j))=\phi((i,j))$ and $\phi'((i,j))= \emptyset$ elsewhere.
\end{definition}
Moreover, for a given $v \in V_X$, we consider the subgraph $G_v$ of
$G_X$ where the vertex set is composed of  $v$ and its neighbors,
commonly denoted $\sigma(v)$, and its set of edges is reduced to $\{(v,i) \in E_X \}$.

\subsection{a-exclusive graphs}
\begin{definition}\label{defExc}
An edge labeled graph $G_X$ is a-exclusive if for all a-edges $(i,j)$, $\phi((i,j))$ contains only $a$.
In other words,
$a \in \phi((i,j)) \implies \neg (\exists b\in[d] {\setminus} \{a\}, b \in \phi((i,j)))$.
\end{definition}
We also can write $a \in \phi(i,j) \implies  \phi((i,j))=\{a\}$.
By extension, we say that a vertex $v \in V_X$ is a-exclusive if the subgraph $G_v$ of $G_X$ is a-exclusive.
%
%

The following special cases are examples of a-exclusive graphs:
\begin{itemize}
\item In the two-dimensional case, if the balls are all disjoint, the induced graph is 1-exclusive and 2-exclusive.
\item If there exists one dimension $a$ on which the projections of all balls are disjoint,  then the induced graph is a-exclusive.
\end{itemize}

\subsection{Chromatic number and independent sets}\label{ssec:chromatic}

%
Two points of the set $X$ (namely $x_i$ and $x_j$ with $i \neq j$) are said to be {\emph{a-dependent}} if
$[x_{ia}- \delta, x_{ia} +\delta] \cap [x_{ja}- \delta, x_{ja} +\delta] \neq \emptyset$.
%
%
%
%
The fundamental problem of partitioning a set of objects into independent subset is named the graph coloring problem.
The Graph Coloring Problem is an NP-hard problem \cite{10.5555/578533} that aims to color the vertices of a graph using the
minimum number of distinct colors, ensuring that adjacent vertices do not share the same color.
Brooks Theorem states that the minimal number of colors (or the {\emph{chromatic number}) $\chi(G)$ is less than  degree of~$G$ plus~$1$.

Some polynomial time heuristic approaches have been developped to
compute this number. Those include greedy algorithm
\cite{diestel2017graph}, \cite{Bollobas1998Modern}, or the DSATUR
algorithm \cite{10.1145/359094.359101}.

The DSATUR algorithm  is a polynomial-time heuristic for the graph
coloring problem based on the concept of {\emph{degree of saturation}}.
At each step, the algorithm selects the vertex most constrained by its
colored neighbors and assigns to it the smallest feasible color.
With a computational complexity bounds of
$O(n^2)$ or $O((n + q) \log{n})$, where $q$ is the number of
edges in the graph, (the latter is better than the former for all but
the densest graphs), DSATUR achieves high-quality coloring.
More precisely, its solutions have a number of colors which is less  than those of the greedy algorithm.
Further details can be found in the cited references. We here only
need to use it as a subroutine and we thus only formalize it as follows.\\
$(Z_i)_{i \in[\chi]}\leftarrow  \textbf{DSATUR}(X, \delta,a)$:  takes as
inputs  a set of $d-$dimensional points $X$ a radius $\delta$ and
optionally $a \in[d]$;
\begin{itemize}
\item If there is no third input, it considers $G_X$ as a simple graph and outputs a set of subset $(Z_i)_{i \in[\chi]}$, such that
$X =\sqcup_{i \in [\chi]} Z_i$ and all the subgraphs $G_{Z_a}$ of $G_X$ are null graphs.
\item If the third input is $a$, it considers the spanning subgraph
  $G^a_X=(V_X,E^a_X)$ and outputs the subsets $(Z_i)_{i
    \in[\chi]}$, s.t. $X =\sqcup_{i \in [\chi]} Z_i$ and all the
  subgraphs $G_{Z_a}$ of $G^a_X$ are null graphs.
\end{itemize}

\section{Adversary Model and Security Definitions}
In this work, we consider two parties : a receiver $\mathcal R$ and a
sender $\mathcal S$.
Highlighting the potential unbalancedness of their resources, they
could respectively be presented as a server and a client.

Our protocols are secure under the \emph{honest-but-curious} adversary model.
This means that the parties must follow the protocol, but are at the
liberty to try to learn some extra information.
In particular no party should learn more
information about the other parties' secrets than what is deductible
from the functionality they are performing.
The associated security proofs we provide use probabilistic polynomial
time simulators following the approach proposed
in~\cite{Lindell2017}.
We first recall the basic notations needed for the proofs.
\begin{itemize}
\item
The view of the receiver (resp. the sender) part during an execution of a protocol $\pi$ is denoted by
$\textbf{view}_{\mathcal R}^{\pi}(X,Y)$ (resp. $\textbf{view}_{\mathcal S}^{\pi}(X,Y)$). Here $X$ corresponds to the input of $\mathcal R$ and $Y$ is the input of $\mathcal S$.
\item
The output of the receiver (resp. the sender)  during an execution of
$\pi$ on $(X,Y)$ is denoted by $\textbf{output}^{\pi}_{\mathcal
  R}(X,Y)$ (resp. by $\textbf{output}^{\pi}_{\mathcal R}(X,Y)$).
\item A receiver (resp. sender) simulator for the protocol $\pi$ is a polynomial time algorithm denoted
$\textbf{Sim}^{\pi}_{\mathcal R}(X,O_{\mathcal R})$ where again $X$ is the input of $\mathcal R$ and
if $O_{\mathcal R}=\textbf{output}^{\pi}_{\mathcal R}(X,Y)$ then we
should have:
$\textbf{Sim}^{\pi}_{\mathcal R}(X,O_{\mathcal R}) \equiv_c\textbf{view}_{\mathcal R}^{\pi}(X,Y).$
A sender simulator is denoted equivalently $\textbf{Sim}^{\pi}_{\mathcal S}(X,O_{\mathcal S})$ and:
$$\textbf{Sim}^{\pi}_{\mathcal S}(Y,\textbf{output}^{\pi}_{\mathcal S}(X,Y)  ) \equiv_c\textbf{view}_{\mathcal R}^{\pi}(X,Y).$$
\end{itemize}

\subsection{Linearly and Fully homomorphic encryption}
A Linear Homomorphic Encryption Scheme (LHE) is a  public-key encryption scheme from $\mathcal V$ to $\mathcal C$ such that
for $(pk,sk)\leftarrow\textbf{E.\SETUP}(\kappa)$, a key pair for a
security parameter $\kappa$, we denote by $\textbf{Enc}_{pk}(m) \in
\mathcal C$ the encryption of $m \in {\mathcal V}$.
The decryption $\textbf{Dec}_{sk}$ must satisfy:
%
\begin{itemize}
\item \textit{homomorphic addition $\textbf{+}$: } $\textbf{Dec}_{sk}(\textbf{Enc}_{pk}(m_1)\textbf{+}\textbf{Enc}_{pk}(m_2))=m_1+m_2$.
\item \textit{cleartext-ciphertext product $\mathbf{\ltimes}$: } $\textbf{Dec}_{sk_L}(m_1\mathbf{\ltimes} \textbf{Enc}_{pk}(m_2))=m_1m_2$.
\end{itemize}

An LHE is a Fully Homomorphic Encryption Scheme (FHE)
if it additionally satisfies:
\begin{itemize}

\item \textit{homomorphic product $\mathbf{\times}$: } $\textbf{Dec}_{sk}(\textbf{Enc}_{pk}(m_1)\mathbf{\times}\textbf{Enc}_{pk}(m_2))=m_1m_2$.

\end{itemize}

In the following, the LHE and FHE schemes we will consider will
satisfy ciphertext indistinguishability under chosen plaintext attacks
(IND-CPA).
 In each FPSU protocols we will propose, we will use two encryptions:
 a FHE for the PIR (Private Information Retrieval) subroutine (see
 \cref{ssec:PIR}) and another encryption, LHE or FHE, depending on the
 protocol's constraints, for the main FPSU protocol.
We will denote by  $\widetilde{x}$ a
value encrypted using fully homomorphic encryption (FHE) under the PIR
subroutine, for instance following \cite{Angel2025}.
There the Learning with Errors (LWE) security assumption is the
central building block of FHE construction
\cite{Albrecht:2015:Concrete}.

We will denote by $\widehat{x}$ a
value encrypted using the second encryption (LHE or FHE). This
encryption scheme will be specified for each protocol.
For instance, we will require that the LHE scheme is secure under the
decisional composite residuosity assumption.
Again we require that the FHE is secure under the LWE
hypothesis.
%

\subsection{Private Information Retrieval}\label{ssec:PIR}
As a subroutine, we need a private information retrieval (PIR)
protocol.
This is a protocol that allows a sender (or a client) to retrieve
mutiple entries from a large
database owned by a receiver (or a server), without revealing what
item is retrieved.

With a database of size $N$, the server time complexity of a PIR
protocol for a single query is of the order of $O(N)$ and this cannot
be reduced with stateless protocols~\cite{Hazay:2023:Optimal}.
(Informally, as the desired index is hidden from the server, for the
query to remain private, the server/receiver has thus to perform some
blind computations on the entire database.)
Then for multiple ($M$) queries, in order to avoid an overall $O(MN)$
server time complexity bound, batchPIR approaches have been developped.

The most recent BatchPIR protocols (namely SealPIR \cite{Angel2025},
OnionPIR \cite{10.1145/3460120.3485381}, Vectorized BatchPIR
\cite{10179329,8401902}, etc.) use buckets to amortize the receiver
computation cost. The overall asymptotic complexity bound on the
receiver side becomes usually $O(N \log N)$.
(Angel \cite{Angel2025}, for instance, uses batch codes for batch
PIR~\cite{10.1145/1007352.1007396}.)
Similarly, the information of $M$ queries, among $N$ indices, must be
of order at least $O(M\log{N})$. Therefore there is a communication
lower bound of that order, even with
preprocessing~\cite{10.1007/978-3-031-30545-0_18}.
Then some optimizations can be done on the
number of round or of ciphertexts: for instance, the communication is
of $2$ ciphertexts per query in \cite{Angel2025}, while Mughees et
al. \cite{10179329,8401902} use RLWE-based SHE schemes to batch
several queries into a single ciphertext.

\begin{figure}[htbp]
    \centering
    \fbox{
        \begin{minipage}{.9\linewidth}
		\centering\textbf{$\mathcal{F}_{\text{PIR}}$: Private Information Retrieval Functionality} \\[-0.3em]
		\rule{\linewidth}{0.25pt} \\[0.25em]

		\raggedright
		\textbf{Parameters:} Security parameters $\kappa$ and $\lambda$.
\textbf{Inputs}
            \begin{itemize}
                \item $\mathcal{R}$ inputs a vector ${K} = \{k_1, \dots, k_{N}\} \subset \mathcal{K}$;
                \item $\mathcal{S}$ inputs some distinct indices ${U} = \{u_1, \dots,u_M \} \subset \mathbb{Z}$ s.t $\forall i \in [M], 1 \leq u_i \leq N$;
                          \end{itemize}

		\textbf{Outputs}
			\begin{itemize}

                \item $\mathcal{R}$  has no outputs
           and learns $M$ but  nothing  about the values of the indices
	\item  $\mathcal{S}$ outputs $K[U]=\{k_{u_1}, \dots, k_{u_{M}}\}$
			\end{itemize}
		\end{minipage}
	}
	\caption{Ideal functionality $\mathcal{F}_{\text{PIR}}$}
    \label{fig:FunctionalityPIR}
\end{figure}

To abstract the  PIR protocol \cref{fig:FunctionalityPIR}, we define the
supporting four procedures involving the two parties. Here the
receiver (viewed as a server) owns a vector $K \in {\mathcal K}^{N}$
and the sender (viewed as a client) owns an index set $U$.

\begin{itemize}
\item
 $ (O^{PIR}_{\mathcal R}, O^{PIR}_{\mathcal S}) \leftarrow
 \textbf{PIR.\SETUP}(\kappa,\lambda)$ is the setup routine that uses as
 inputs  a computational security parameter $\kappa$ and a statistical
 security parameter $\lambda$. The sender gets $O^{PIR}_{\mathcal S}$
 containing a pair $(pk_1,sk_1)$ of public and secret keys of an
 encryption scheme, the receiver gets the parameters
 $O^{PIR}_{\mathcal R}$ containing only $pk_1$.
\item
 $(\tilde I,J) \leftarrow \textbf{PIR.\QUERY}(O^{PIR}_{\mathcal S},U)$:
 This sender's routine inputs $O^{PIR}_{\mathcal S}$ and a sequence of
 distinct indices $U \subset [N]$. The output is the sequence of
 encrypted indices $\tilde I$, to be sent to ${\mathcal R}$, and a
 private state $J$ ($J$ is kept secret by ${\mathcal S}$). As $\tilde
 I$ (resp.  $J$ ) is
 the first (resp. second) output of $\textbf{PIR.\QUERY}$, we will also
 denote it by $\textbf{PIR.\QUERY}(O^{PIR}_{\mathcal S},U)[1]$
 (resp. $\textbf{PIR.\QUERY}(O^{PIR}_{\mathcal S},U)[2]$).
\item
 $\tilde r \leftarrow \textbf{PIR.\RESP}(O^{PIR}_{\mathcal R}, K,\tilde I)$: This receiver's routine inputs a vector $K$ and a query $\tilde I$ and uses  $O^{PIR}_{\mathcal R}$ to output a response $\tilde r$ to be sent to the sender. In order to amortize the costs, the database $K$ has been organized into $b$ buckets using $z$ hash functions.
\item
 $ (w_1,\cdots, w_M) \leftarrow \textbf{PIR.\DECOD}(O^{PIR}_{\mathcal S},J,\tilde r)$: This sender's routine inputs $\tilde r$ and $J$ and uses $O^{PIR}_{\mathcal S}$ to  output a list of elements such that:
\begin{multline}
\textbf{PIR.\DECOD}(O^{PIR}_{\mathcal S},\textbf{PIR.\QUERY}(O^{PIR}_{\mathcal
  S},U)[2],
\textbf{PIR.\RESP}(K,\textbf{PIR.\QUERY}(O^{PIR}_{\mathcal S},U)[1]))=K[U]
\end{multline}
\end{itemize}

The expected security of such a protocol is as follows:
\begin{itemize}
\item {\bf Client's privacy}: the server learns nothing about the
(batch of) indices requested by the client;
\item {\bf Server's privacy}: the client learns nothing about the not
  requested entries;
\item {\bf Correctness}: if the client and the server correctly execute
the protocol, then the client recovers the entries associated to the
requested (batch of) entries.
\end{itemize}

In the following, we consider PIR protocols which are correct and
secure under the honest-but-curious model.
In particular, under the decision-LWE problem of the FHE \cite{Albrecht:2015:Concrete}, the considered PIR protocol \cite{Angel2025} satisfies, for all $K \in {\mathcal K}^{N}$ and for all distinct sequence of indices  $U \subset [N]$, $\textbf{view}_{\mathcal R}^{PIR}(K,U)\equiv_c \textbf{Sim}_{\mathcal R}^{PIR}(K,\bot)$ and $\textbf{view}_{\mathcal S}^{PIR}(K,U) \equiv_c \textbf{Sim}_{\mathcal S}^{PIR}(U,K[U])$.
Finally, \cite{Angel2025} provides a secure instantiation of this
functionality under the LWE assumption.

\subsection{Split Oblivious Key-Value Store}\label{ssec:SOKVS}
A variant of the PIR functionality considers a set of key-value pairs
${KC}=\{(k_1,c_1),\cdots,(k_N,c_N)\}$ where $k_i (i \in[N])$ belongs
to a set of keys  $ \mathcal{K}$ and  $c_i (i \in[N])$ belongs to a
set $\mathcal{C}$. The set is {\emph{well-formed}} if
$\forall 1 \leq i<j\leq N, k_i \neq k_j$.

A key-value store is a map $f$ such that  $f(k)=v$  for all $(k,v)$ belonging to the set of key-value pairs.
Moreover the  key-value store is oblivious (OKVS) if $f$ can be
computed without releasing any information on the values of the keys.
Since $f$ can be large, a sparse OKVS approach has been proposed by
\cite{298226} in order to reduce the communication cost of retrievals
from $O(N)$ to $O(\log N)$.
We further need to refine this to a {\emph{Split OKVS}} (SOKVS)
functionality where the (encoded and encrypted) store is split in two
parts: the first part is kept by the owner, while a (small) part is
provided to the client.
The formalization of this new variant is depicted in
\cref{fig:FunctionalityOKVS}.
\begin{figure}[htbp]
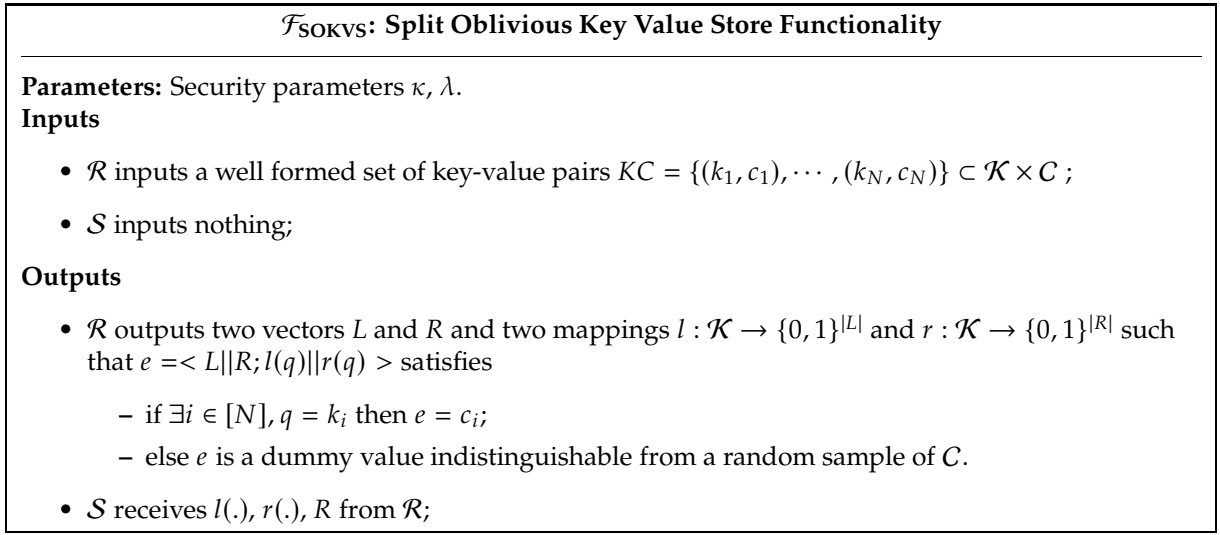

    \centering
    \fbox{
        \begin{minipage}{.95\linewidth}
		\centering\textbf{$\mathcal{F}_{\text{SOKVS}}$: Split Oblivious Key Value  Store Functionality} \\[-0.3em]
		\rule{\linewidth}{0.25pt} \\[0.25em]

		\raggedright
		\textbf{Parameters:} Security parameters $\kappa$, $\lambda$.\\
\textbf{Inputs}
            \begin{itemize}
                \item $\mathcal{R}$ inputs a well formed set of key-value pairs ${KC} = \{(k_1,c_1), \cdots, (k_N,c_N)\} \subset \mathcal{K} \times \mathcal{C}$ ;
                \item $\mathcal{S}$ inputs nothing;
                          \end{itemize}

		\textbf{Outputs}
			\begin{itemize}

                \item $\mathcal{R}$  outputs two vectors $L$ and $R$ and two mappings $l: {\mathcal K} \rightarrow \{0,1\}^{|L|} $ and  $ r: {\mathcal K} \rightarrow \{0,1\}^{|R|}$  such that $e=<L||R;l(q)||r(q)>$ satisfies
\begin{itemize}
\item if $\exists i \in[N], q=k_i$ then $e=c_i$;
\item else $e$ is a dummy value indistinguishable from a random sample
  of $\mathcal{C}$.
\end{itemize}
	\item  $\mathcal{S}$ receives $l(.)$, $r(.)$, $R$ from $\mathcal R$;
			\end{itemize}
		\end{minipage}
	}
	\caption{Ideal functionality $\mathcal{F}_{\text{SOKVS}}$}
    \label{fig:FunctionalityOKVS}
\end{figure}

The formal routines of the SOKVS are as follows:
\begin{itemize}
\item
$(O^{SOKVS}_{\mathcal R},O^{SOKVS}_{\mathcal S}) \leftarrow {\textbf{SOKVS.\SETUP}}(\kappa, \lambda,N)$ is the setup routine that uses as inputs a computational security parameter $\kappa$ and a statistical security parameter $\lambda$. The receiver gets the parameters $O^{SOKVS}_{\mathcal R}$ and the sender gets $O^{SOKVS}_{\mathcal S}$. Both of them contains the random mappings $r(.)$ and $l(.)$.
\item
$(L,R) \leftarrow {\textbf{SOKVS.\ENCOD}}(O^{SOKVS}_{\mathcal R},KC)$: This receiver's routine inputs a set of key-value pairs and a setup $O^{SOKVS}_{\mathcal R}$ and outputs the vectors $R$ and $L$
 such that $<l(q)||r(q),L||R>=e$ where
\begin{itemize}
\item if $\exists i \in[N], q=k_i$ then $e=c_i$;
\item else $e$ is a dummy value indistinguishable from a random sample
  of $\mathcal{C}$.
\end{itemize}
\end{itemize}

To preserve the low communication cost of the OKVS, these constraints
must be met:
\begin{itemize}
\item The output of {\textbf{SOKVS.\ENCOD}} satisfies
$|R|=O(log N)$ and $|L|=O(N)$;
\item $l: {\mathcal K} \rightarrow \{0,1\}^{|L|}$ outputs a sparse
  binary vector which weight is bounded by a constant $\alpha$.
\end{itemize}
$R, l(.), r(.)$ are sent to the client, while the server keeps the
large vector $L$.
By construction, $l(q)$ is a sparse sub-vector of weight not greater
than $\alpha$.
As $L$ is unknown to the sender, the latter is not able to retrieve
any value directly.
The idea is that A PIR protocol has to be activated to obtain the
$\alpha$ values required to compute $<l(q),L>$ for a given $q$, and
then recover the actually sought elements.

In the following, we consider SOKVS protocols (such as the one
designed by Hao et al. \cite{298226}) which are correct and secure
under the {\emph{honest-but-curious model}}.
In particular,
for any well-formed set of pairs $\{(k_i,c_i)_{i \in [N]} \}$ and
$\{(k'_i,c_i)_{i \in [N]} \}$, we have the
 {\emph{indistinguishability of the encoding}}:
\begin{equation}
\textbf{SOKVS.\ENCOD}(\{k_i,c_i\})~{\equiv_c}~\textbf{SOKVS.\ENCOD}(\{k'_i,c_i\}),i{\in}[N]\!
\end{equation}

Moreover, for all distinct $\{k_1,\cdots,k_N \} \subset {\mathcal K}$,
if $\{c_1,...c_N\} \subset {\mathcal C}$ is sampled uniformly at
random from $\mathcal C$ and if $\textbf{SOKVS.\ENCOD}$ succeeds to build
$L$ and $R$, then $L||R$ is statistically {\emph{indistinguishable
    from a uniformly random vector}} of ${\mathcal C}^{|L|+|R|}$:
\begin{equation}
\textbf{SOKVS.\ENCOD}(\{(k_i,c_i)_{i \in[N]}  \}) {\equiv_c} Rand( {\mathcal
  C})^{|L|+|R|}
\end{equation}

Then, we also have the following {\emph{double obliviousness}}:
\begin{equation}
\textbf{SOKVS.\ENCOD}(\{k_i,c_i\})~{\equiv_c}~\textbf{SOKVS.\ENCOD}(\{k'_i,c'_i\}),i{\in}[N]\!
\end{equation}

Finally, \cite{DBLP:conf/ccs/RaghuramanR22} provides a secure
instantiation of this functionality under the so-called Alekhnovich
assumption~\cite{Alekhnovich:2011:Average,Malkin:2021:SilverVOLE,Rindal:2021:VOLEPSI,DBLP:conf/ccs/RaghuramanR22}.


\section{Oblivious Key Homomorphic Encryption Retrieval}
Now, we propose a new functionality named Oblivious Key Homomorphic
Encryption Retrieval (OKHER). There the server holds a  set of
key-value pairs $KV=\{  (k_i,v_i)_{i\in[N]}  \} \subset \mathcal K
\times \mathcal C$ where $N$ is large and the clients holds a set of
$m$ queries $Q=\{(q_i)_{i\in[m]} \} \subset \mathcal K$ such that
$m<<N$.
At the end of the OKHER's execution, the server will learn nothing
about the client queries.
 For all her queries $q_i$, the client will receive a
 {\emph{ciphered}} value
 equal to $Enc(v_j)$, if there exists a $j$ such that
 $q_i=k_j$. Otherwise, he will receive a dummy value belonging to the
 ciphertext space, indistinguishable from a randomly sampled value.
The idea is that the encryption is {\emph{homomorphic}}, so that while
encrypted, the value can be used afterwards.

\subsection{Formalization}
The functionality of an OKHER is presented in \cref{fig:FunctionalityOKHER}:

\begin{figure}[htbp]
    \centering
    \fbox{
        \begin{minipage}{1\linewidth}
		\centering\textbf{$\mathcal{F}_{\text{OKHER}}$: Oblivious Key Homomorphic Encryption Retrieval Functionality} \\[-0.3em]
		\rule{\linewidth}{0.25pt} \\[0.25em]

		\raggedright
		\textbf{Parameters:} Security parameters $\kappa$, $\lambda$, an encryption scheme $(Enc,Dec)$\\
\textbf{Inputs}
            \begin{itemize}
                \item $\mathcal{R}$ inputs a set of key-value pairs ${KV} = \{(k_1,v_1), \cdots, (k_N,v_N)\} \subset \mathcal{K} \times \mathcal{V}$;
                \item $\mathcal{S}$ inputs a set of queries $Q=\{q_1,\cdots, q_m\} \subset \mathcal{K}$;
                          \end{itemize}

		\textbf{Outputs}
			\begin{itemize}

                \item $\mathcal{R}$  outputs $(pk,sk)$ the pair of public and secret keys of the encryption scheme $(Enc,Dec)$,
\item  $\mathcal{S}$ outputs $pk$ and $E=\{e_1, \cdots, e_m \} \subset \mathcal{C}$
such that for all $j \in[m]$,
\begin{itemize}
\item if $\exists i \in[N], q_j=k_i$ then $e_j=Enc_{pk}(v_i)$;
\item else $e_j$ is a value viewed as a random over $\mathcal{C}$.
\end{itemize}
			\end{itemize}
		\end{minipage}
	}
	\caption{Ideal functionality $\mathcal{F}_{\text{OKHER}}$ }
    \label{fig:FunctionalityOKHER}
\end{figure}

We now propose a formal definition
of an OKHER protocol divided into  four algorithms, namely
$\textbf{OKHER.\SETUP}$, $\textbf{OKHER.\QUERY}$, $\textbf{OKHER.\RESP}$,
$\textbf{OKHER.\DECOD}$. The sender $\mathcal{S}$ owns a set of key-value
pairs $KV \subset \mathcal{K}\times \mathcal{V}$ and the receiver
$\mathcal{R}$ owns a set of queries $Q \subset \mathcal K$.
The protocol needs as sub-protocols a PIR and a SOKVS as defined in~\cref{ssec:PIR,ssec:SOKVS}.
\begin{itemize}
\item
 $ (O^{OKHER}_{\mathcal R}, O^{OKHER}_{\mathcal S}) \leftarrow
 \textbf{OKHER.\SETUP}(\kappa,\lambda, (Enc,Dec),N)$ is the setup routine
 that uses as inputs  a computational security parameter $\kappa$, a
 statistical security parameter $\lambda$, an encryption scheme
 $(Enc,Dec)$ and the size of the database $N$. The receiver gets the
 parameters $O^{OKHER}_{\mathcal
   R}=\{pk,sk,O^{\mathcal{SOKVS}}_{\mathcal
   R},O^{\mathcal{PIR}}_{\mathcal R} \}$ and the sender gets
 $O^{OKHER}_{\mathcal S}=\{pk,O^{\mathcal{SOKVS}}_{\mathcal
   S},O^{\mathcal{PIR}}_{\mathcal S} \}$.
\item
 $((\tilde I_i)_{i \in[M]},J, \beta) \leftarrow
 \textbf{OKHER.\QUERY}(O^{OKHER}_{\mathcal S},Q)$ : This sender's routine
 inputs a set of queries $Q \in {\mathcal{K}}^m$, and outputs
 encrypted indices $ (\tilde I_i)_{i \in [M]}$, to be sent to
 ${\mathcal R}$, as well as a private state $J$ (to be used in
 \textbf{PIR.\QUERY}) and $\beta$ (for the ${\textbf{Merge}}$ function,
 see below).
\item
 $(R,(\tilde \rho)_{i \in[M]}) \leftarrow \textbf{OKHER.\RESP}(O^{OKHER}_{\mathcal R},KV,\tilde I)$: This receiver's routine inputs the set of key-value pairs $KV$,  a query $(\tilde I_i)_{i \in[M]}$ and uses  $O^{OKHER}_{\mathcal R}$ to output  the right part of the decoding vector of the SOKVS and a response $(\tilde \rho)$ to be sent to the sender.
\item
 $(e_1,\cdots, e_m) \leftarrow \textbf{OKHER.\DECOD}(O^{OKHER}_{\mathcal S},R, J,\beta, \tilde \rho)$ : This sender's routine inputs $O^{OKHER}_{\mathcal S},R,J,\beta, \tilde \rho$ and  outputs a list of elements
such that for all $j \in[m]$,
\begin{itemize}
\item if $\exists i \in[N], q_j=k_i$ then $e_j=Enc(v_i)$;
\item else $e_j$ is a dummy value indistinguishable from a random
  sample of $\mathcal{C}$.
\end{itemize}

\end{itemize}

The encryption scheme is expected to come with a setup function
{\textbf{E.\SETUP}} building the associated keys from a security
parameter.

Moreover, $\textbf{OKHER.\QUERY} $ needs two more utility functions:
\begin{itemize}
\item $\mu \leftarrow \textbf{NZIndices}(V)$ is a function which extract from a boolean vector $V$ the indices of non zero coefficients and outputs a vector of indices $\mu$.
\item $(U,\beta) \leftarrow \textbf{Merge}( (\mu_i)_{i \in [m]})$
  concatenates several vectors by removing duplicate
  values. It inputs $m$ vectors $(\mu_i)_{i\in[m]} $ of size $\alpha$
  ( $\mu_i \in [n]^\alpha$ ) and outputs a vector of length $M \leq
  \alpha m$ of distinct values $U$ and a $[m] \times [\alpha]$ matrix
  $\beta$ such that the row $\beta_i$ contains the $\alpha$ indices of
  $U$, included in $[1..M]$, and corresponding to the $i$-th query.
i.e.  $j \in \beta[i] \Leftrightarrow \exists l \in[\alpha], \mu_{il}=U_j$.
\end{itemize}

In \cref{proto:OKHER} we show an instantiation of the OKHER protocol,
constructed from the SOKVS and PIR building blocks.

\begin{figure*}[htbp]\centering
\caption{Oblivious Key Homomorphic Encryption Retrieval (OKHER)} \label{proto:OKHER}
    \begin{tabular}{|l c l|}
\hline
    \Huge{$\mathcal{R}$} &&\Huge{$\mathcal{S}$} \\
    $KV=\{ (k_i,v_i)_{i\in[N]}\}\subset\mathcal{K}\times \mathcal{V}$, $\kappa$, $\lambda, (Enc,Dec)$ & & ${Q}=\{q_1,...,q_m\}\subset \mathcal{K}$ \\ \hline
&\textbf{OKHER.\SETUP} (Alg. \ref{algo:OKHER:Su}):&\\
   $(pk,sk)\leftarrow \textbf{E.\SETUP}(\kappa )$&&\\

$(O^{SOKVS}_{\mathcal{R}},O^{SOKVS}_{\mathcal{S}}) \leftarrow \textbf{SOKVS.\SETUP}(\kappa, \lambda,N)$&$\xleftarrow{O^{PIR}_{\mathcal{R}}}$ $\xrightarrow{pk,O^{SOKVS}_{\mathcal{S}}}$&$ (O^{PIR}_{\mathcal{R}},O^{PIR}_{\mathcal{S}})\leftarrow   \textbf{PIR.\SETUP}(\kappa,\lambda)$\\
\hline
&\textbf{OKHER.\QUERY} (Alg. \ref{algo:OKHER:Qr}):&\\
&& Extract $l(.)$ from $O^{SOKVS}_{\mathcal{S}}$\\
&&   $(\mu_i)_{i\in[m]} \leftarrow \textbf{NZIndices}(l(q_i))_{i\in[m]}$\\
&& $(U,\beta) \leftarrow \textbf{Merge}((\mu_i)_{i\in[m]})$\\
&$\xleftarrow{\tilde I}$&$ (\tilde I, J)  \leftarrow \textbf{PIR.\QUERY}(O^{PIR}_{\mathcal{S}},U)$\\
\hline
&\textbf{OKHER.\RESP} (Alg. \ref{algo:OKHER:Rs}): &\\
$KC=\{  (k_i,Enc_{pk}(v_i))_{i\in[N] }  \}$&&\\
$(L, R) \leftarrow \textbf{SOKVS.\ENCOD}(O^{SOKVS}_{\mathcal{R}},KC)$&&\\
$\tilde \rho \leftarrow\textbf{PIR.\RESP}(O^{PIR}_{\mathcal{R}},L,\tilde I) $&$\xrightarrow{R,\tilde \rho}$&\\

\hline
&\textbf{OKHER.\DECOD} (Alg. \ref{algo:OKHER:Dc}): &\\

&&$W \leftarrow \textbf{PIR.\DECOD}(O^{PIR}_{\mathcal{S}}, J, \tilde \rho)$\\
&& Extract $r(.)$ from $O^{SOKVS}_{\mathcal{S}}$\\
 && $( e_i)_{i\in[m]} \leftarrow (<r(q_i),R>+\sum_{j=1}^\alpha w_{\beta_{ij}})_{i\in[m]}$\\
\hline
    \end{tabular}
    \end{figure*}

\subsection{OKHER correctness and security}
\begin{theorem}\label{thm:OKHER}
Given an IND-CPA encryption scheme and  correct and secure
\textbf{PIR} and \textbf{SOKVS} schemes, with a probability of success
of $(1- 2^{-\lambda})$, the protocol defined in \cref{proto:OKHER} and based on the 4 algorithms
$\textbf{OKHER.\SETUP}$, $\textbf{OKHER.\QUERY},$ $\textbf{OKHER.\RESP},$  $\textbf{OKHER.\DECOD}$  is correct and secure under honest-but-curious model with the same probability of success $(1- 2^{-\lambda})$.
\end{theorem}

The proof is similar to those of \cite[\S~4]{298226}. The main
difference is that the OPRF in \cite{298226} is replaced by an
encryption. The former is indeed required to guarantee that the
receiver's values are indistinguishable from a uniformly sampled
distribution. In OKHER this is ensured by the IND-CPA assumption on
the encryption scheme.
More precisely, for all well-formed  set of key-value pairs $KV$ and
for all set of queries $Q$, we have:
\begin{align}
\textbf{view}^{OKHER}_{\mathcal{S}}(KV,Q)&\equiv_c\textbf{Sim}^{OKHER}_{\mathcal{S}}(Q,output_S(KV,Q))\\
\textbf{view}^{OKHER}_{\mathcal{R}}(KV,Q)&\equiv_c\textbf{Sim}^{OKHER}_{\mathcal{R}}(KV,
(pk,sk) \})
\end{align}
The details of the proof are given in \cref{OKHERproof}.


The volume of communication is in two parts:
sending $R$ requires $\lambda+O(\log n)$ communications, while
$m\alpha$ PIR queries on $L$ need $O(m \log n)$,
using Angel et al. PIR scheme~\cite{Angel2025}.
The overall asymptotic complexity bounds are presented
in~\cref{tab:OKHERComplexity}:
\begin{table}[htbp]
    \centering
     \caption{Complexity bounds of the OKHER protocol: The receiver has a set of $N$  key-value pairs and the sender has $m$ queries}
    \label{tab:OKHERComplexity}
    \begin{tabular}{l|ccc}
\toprule
        \textbf{Procedure} & \textbf{Receiver} & \textbf{Sender} & \textbf{Comm.} \\
\midrule
        OKVS \cite{DBLP:conf/ccs/RaghuramanR22} & $\tilde{ \mathcal{O}}( N)$ & $\mathcal{O}(\log (N))$ & $\mathcal{O}(\log (N))$ \\
        PIR \cite{Angel2025} &  $\tilde{\mathcal{O}}(N) $ & $\mathcal{O}(m\log (N))$ & $\mathcal{O}(m\log (N))$ \\
        Total & $\tilde{\mathcal{O}}( N)$ & $\mathcal{O}(m \log (N))$ & $\mathcal{O}(m \log (N))$ \\
\bottomrule
    \end{tabular}

\end{table}

\begin{remark} We will denote an \textbf{OKHER} protocol that uses
  a linearly homomorphic encryption by \textbf{ LOKHER},
  and a protocol that uses a
  fully homomorphic encryption by \textbf{FOKHER}.
\end{remark}
 In a practical point of view and in an honest-but-curious model, \cref{thm:OKHER} becomes:
\begin{itemize}
\item
LOKHER implemented with SealPIR \cite{Angel2025}, Paillier cryptosystem \cite{10.5555/1756123.1756146} and OKVS \cite{DBLP:conf/ccs/RaghuramanR22} 
is correct and secure under the LWE hypothesis of security, the
composite residuosity assumption and the Alekhnovich
assumption~\cite{Alekhnovich:2011:Average,Malkin:2021:SilverVOLE,Rindal:2021:VOLEPSI,DBLP:conf/ccs/RaghuramanR22},
with the probability of success $(1- 2^{-\lambda})$.

\item 
FOKHER implemented with SealPIR~\cite{Angel2025}, the BGV cryptosystem
\cite{Brakerski2011FullyHE} and the OKVS of~\cite{DBLP:conf/ccs/RaghuramanR22}
 is correct and secure under the LWE hypothesis of security and and the Alekhnovich
assumption, with the
 probability of success $(1- 2^{-\lambda})$.
\end{itemize}

\subsection{Formal Algorithms}
We now provide the detailed formalization of our OKHER protocol, as
depicted in~\cref{proto:OKHER}.
\begin{algorithm}[htbp]
\caption{$\textbf{OKHER.\SETUP}(\kappa, \lambda, (Enc,Dec),N)$}\label{algo:OKHER:Su}
\begin{flushleft}
\textbf{Input:} A security parameter $\kappa$,  an encryption scheme
$(Enc,Dec)$ over $\mathcal V$ and $\mathcal C$, with setup \textbf{E.\SETUP}.\\
    \textbf{Outputs:} A set of of secret and share keys $O^{OKHER}_\mathcal{R}$ and $O^{OKHER}_\mathcal{S}$.
\end{flushleft}
\begin{algorithmic}[1]
\State $\mathcal{R}$: compute $(pk,sk)\leftarrow \textbf{E.\SETUP}(\kappa, \lambda, (Enc,Dec))$;


\State $\mathcal{R}$: compute $(O^{SOKVS}_{\mathcal{R}},O^{SOKVS}_{\mathcal{S}}) \leftarrow \textbf{SOKVS.\SETUP}(\kappa, \lambda, (Enc,Dec),N)$;
\State $\mathcal{S}$: compute $ (O^{PIR}_{\mathcal{R}},O^{PIR}_{\mathcal{S}})\leftarrow   \textbf{PIR.\SETUP}(\kappa,\lambda)$;

\State $\mathcal{R}$: send $ O^{SOKVS}_{\mathcal{S}}, pk$ to the sender;
\State $\mathcal{S}$: send $ O^{PIR}_{\mathcal{R}}$ to the receiver;
\State $\mathcal{R}$: \Return $O^{OKHER}_\mathcal{R}\leftarrow ((pk,sk), O^{PIR}_{\mathcal{R}},O^{SOKVS}_{\mathcal{R}})  $;
\State $\mathcal{S}$: \Return $O^{OKHER}_\mathcal{S}\leftarrow ( pk,O^{PIR}_{\mathcal{S}},O^{SOKVS}_{\mathcal{S}}  )$;
\end{algorithmic}
\end{algorithm}

    \begin{algorithm}[htbp]
        \caption{\textbf{OKHER.\QUERY}$(O^{OKHER}_\mathcal{S}, Q)$}
        \label{algo:OKHER:Qr}
            \begin{flushleft}
            \textbf{Input:} A set of parameters generated previously $O^{OKHER}_\mathcal{S}$ and the elements of the sender set $\textbf{Q} = \{q_1, \dots, q_m\} \subset {\mathcal K}$ \\
            \textbf{Output:} The sequence $(\tilde I_i)_{i \in [M]}$, $J$ and $\beta$ , $\tilde I_i$ is of size $\alpha$.
            \end{flushleft}
        \begin{algorithmic}[1]
\State $\mathcal{S}$: extract $ O^{PIR}_{\mathcal{S}}$ and $l(.)$ from $O^{OKHER}_\mathcal{S}$.
\ForAll{$i \in [m]$}
\State $\mathcal{S}$: compute    $\mu_i \leftarrow \textbf{NZIndices}(l(q_i))$;  \Comment{$\mu_i \in ([N])^{\alpha}$}
\EndFor
\State $\mathcal{S}$: compute $(U,\beta) \leftarrow \textbf{Merge}((\mu_i)_{i\in[m]})$
\State $\mathcal{S}$: compute $ (\tilde I,J) \leftarrow \textbf{PIR.\QUERY}(O^{PIR}_{\mathcal{S}},U)$;

\State $\mathcal{S}$: send $\tilde I$ to $\mathcal{R}$;
\State $\mathcal{S}$: \Return $J, \beta$;
\State $\mathcal{R}$: \Return $\tilde I $.
\end{algorithmic}
    \end{algorithm}


    \begin{algorithm}[htbp]
        \caption{\textbf{OKHER.\RESP}$(O^{OKHER}_\mathcal{R}, KV,\tilde I)$}
        \label{algo:OKHER:Rs}
            \begin{flushleft}
            \textbf{Input:} A set of parameters $O^{OKHER}_\mathcal{R}$and the set of masked queries $(\tilde I_i)_{i \in [M]}$. \\
            \textbf{Output:} $R$ the right part of the decoding vector of the SOKV and a set of masked responses $(\tilde \rho_i)_{i \in[M]} $.
            \end{flushleft}
        \begin{algorithmic}[1]
\State $\mathcal{R}$: extract $ O^{PIR}_{\mathcal{R}}$, $ O^{SOKVS}_{\mathcal{R}}$ and $pk$ from $O^{OKHER}_\mathcal{R}$;
\State $\mathcal{R}$: $KC=\{  (k_i,Enc_{pk}(v_i))_{i\in[N] }  \}$; \Comment{where $KV=\{ (k_i,v_i)_{i\in[N]}\}$}
\State $\mathcal{R}$: $(L, R) \leftarrow \textbf{SOKVS.\ENCOD}(O^{SOKVS}_{\mathcal{R}},KC)$;
\State $\mathcal{R}$: compute    $\tilde \rho\leftarrow\textbf{PIR.\RESP}(O^{PIR}_{\mathcal{R}},L,\tilde I)$ ;
\State $\mathcal{R}$: send $R$ and $\tilde \rho$  to $\mathcal{S}$;
\State $\mathcal{S}$: \Return $R, \tilde \rho$.
\end{algorithmic}
    \end{algorithm}

    \begin{algorithm}[htbp]
        \caption{\textbf{OKHER.\DECOD}$(O^{OKHER}_\mathcal{S}, R,J, \beta, \tilde \rho)$}
        \label{algo:OKHER:Dc}
            \begin{flushleft}
            \textbf{Input:} A set of parameters $O^{OKHER}_\mathcal{R}$, $R$, the secret data $J$ and $\beta$ and the vector of encrypted responses $\tilde \rho$. \\
            \textbf{Output:} A list $(e_i)_{i \in [m]} $  such that for all $j \in[m]$,
\begin{itemize}
\item if $\exists i \in[N], q_j=k_i$ then $e_j=Enc_{pk}(v_i)$;
\item else $e_j$ is a value viewed as a random over $\mathcal{C}$.
\end{itemize}
            \end{flushleft}
        \begin{algorithmic}[1]
\State $\mathcal{S}$: extract $ O^{PIR}_{\mathcal{S}}$ and $r(.)$  from $O^{SOKVS}_\mathcal{S}$.

\State $\mathcal{S}$: compute  $W\leftarrow \textbf{PIR.\DECOD}(O^{PIR}_{\mathcal{S}}, J, \tilde \rho)$;
\ForAll{$i \in [m]$}
\State $\mathcal{S}$: compute $e_i \leftarrow <r(q_i),R> + \sum_{j=1}^\alpha w_{\beta_{ij}}$; \Comment{$\beta \in [m] \times [\alpha]$}
\EndFor

\State $\mathcal{S}$: \Return $(e_i)_{i \in[m]}$.

\end{algorithmic}
    \end{algorithm}


\section{Fuzzy Private Set Union }
We consider a $d$-dimensional space ${\mathcal K}^d$ and the $L_\infty$ distance.
In this context, the definition of the fuzzy union of two sets is as follows:
\begin{definition}
The fuzzy union of $X = \{x_1, . . . ,x_n\} \subset {\mathcal K}^d$ and  $Y = \{y_i, . . . ,y_m\}\subset {\mathcal K}^d$ with respect to a threshold $\delta$ is the set
$X \cup_F Y= X \sqcup \left\{ y_i \in Y \mid \forall j \in[n],  \text{dist}(x_j, y_i) > \delta \right\}$
\end{definition}

We start by presenting an ideal functionality for FPSU
in~\cref{fig:FunctionalityFPSU}, that allows two parties, a receiver
and a sender,  to compute the fuzzy union of their private sets.
The receiver holds a set $X = \{x_1, . . . ,x_n\} \subset {\mathcal
  K}^d$ and defines a (fuzzy) distance $\delta$; the sender holds a
set $Y = \{y_i, . . . ,y_m\}\subset {\mathcal K}^d$.
The receiver learns the value of a point $y_i$ if and only if it is
not included in a ball of radius $\delta$ centered in any elements of~$X$.


\begin{figure}[htbp]
    \centering
    \fbox{
        \begin{minipage}{.95\linewidth}
		\centering\textbf{$\mathcal{F}_{\text{FPSU}}$: Fuzzy Private Set Union Functionality}  \\[-0.3em]
		\rule{\linewidth}{0.25pt} \\[0.25em]

		\raggedright
		\textbf{Parameters:} Security parameters $\kappa$, $\lambda$, An encryption scheme $(Enc,Dec)$,a set $\mathcal{K}$, a dimension space $d$, the $L_\infty$ distance function $\text{dist}(.,.)$ over $\mathcal K$.\\
\textbf{Inputs}
            \begin{itemize}
                \item $\mathcal{R}$ inputs $\textbf{X} = \{x_1, \dots, x_n\} \subset \mathcal{K}^d$ and a radius $\delta$;
                \item $\mathcal{S}$ inputs $\textbf{Y} = \{y_1, \dots,y_m\} \subset \mathcal{K}^d$;

            \end{itemize}

		\textbf{Outputs}
			\begin{itemize}

                \item $\mathcal{R}$ outputs
            $
             X \cup_F Y=X \bigsqcup \left\{ y_i \in Y \mid \forall j \in[n], \text{dist}(x_j, y_i) > \delta \right\}
            $, he learns nothing except $m$ and $|\textbf{Y} \cap (\cup_{i=1}^n\mathcal{B}_\delta (x_i))|$.
	\item  $\mathcal{S}$ has no output and learns nothing about the receiver's input.
			\end{itemize}
		\end{minipage}
	}
	\caption{Ideal functionality $\mathcal{F}_{\text{FPSU}}$ for fuzzy private set union.}
    \label{fig:FunctionalityFPSU}
\end{figure}

In order to build a protocol for this functionality, we use the OKHER
subroutine.
The first idea to avoid a $\delta ^d$ factor in the computation complexity,
,as in \cite{10.1007/978-3-031-58740-5_12}, is to use $d$ projections.
 Unfortunately, the construction of well-formed set of key-value pairs
 in that general setting can be difficult and we need to consider
 different cases.
If the balls are disjoint (meaning that some but not all of their axis
projection can  overlap), \cite{10.1007/978-3-031-15802-5_12} propose
a method based on block structures that identify one ball knowing the
projections in $d$ dimensions. This technique multiplies the complexity
cost by $2^d$. To be more efficient, two methods can be considered:
\begin{itemize}
\item Force the keys to be all different
\item Build a subset of disjoint data along the axes, such that the
  disjoint union of these subsets describes $X$ as a whole.
\end{itemize}
We explore these possibilities in the following sections.

\subsection{FPSU on null graph}\label{NGSection}
In this subsection, we focus on a receiver set $X= \{x1,\ldots,x_n\}$
and a fuzzy threshold $\delta$ such that the induced graph is
null.
In other words, the balls ${\mathcal B}_\delta(x_i)$ are axes disjoints.
For this, we build a Fuzzy private set union protocol (FPSU)
named {\textbf \NGFPSU}. This protocol uses an IND-CPA linear
homomorphic encryption scheme over $\mathcal{V}$: $(Enc_{pk},Dec_{sk})$
and {\textbf{LOKHER}} as the main building block. The overall description
is available in \cref{proto:NGFPSU} and the details of the resulting
Setup algorithm \textbf{FPSU.\SETUP },
Query algorithm \textbf{\NGFPSU.\QUERY},
Response algorithm \textbf{\NGFPSU.\RESP},
Decode algorithm \textbf{\NGFPSU.\DECOD}
and Union computation \textbf{FPSU.\UNION},
are available in Algorithms~\ref{algo:FPSU:Su}, \ref{algo:NGFPSU:Qr},
\ref{algo:NGFPSU:Rs}, \ref{algo:NGFPSU:Dc} and  \ref{algo:FPSU:Un}.

\begin{figure*}[htbp]
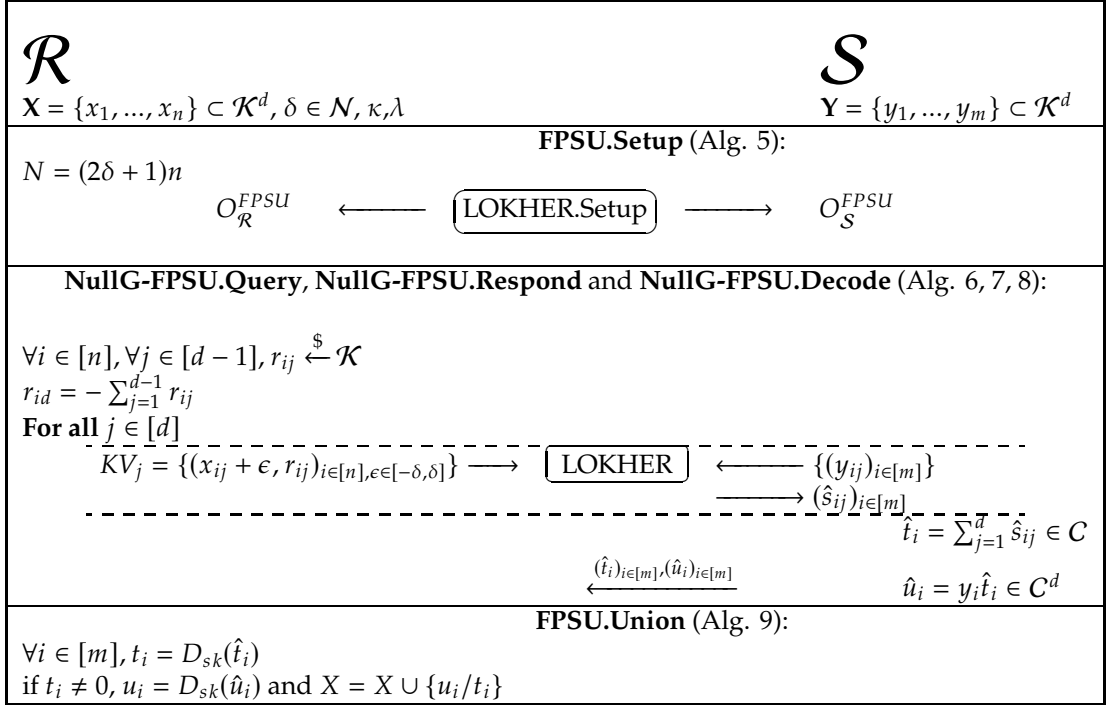
\centering
\caption{ \NGFPSU:  FPSU on null graph  protocol } \label{proto:NGFPSU}
    \hspace*{-1cm}\begin{tabular}{|l c l|}
\hline
&&\\
    \Huge{$\mathcal{R}$} &&\Huge{$\mathcal{S}$} \\
    $\textbf{X}=\{ x_1,...,x_n\}\subset\mathcal{K}^d$, $\delta \in \mathcal{N}$, $\kappa$,$\lambda$ & & $\textbf{Y}=\{y_1,...,y_m\}\subset \mathcal{K}^d$ \\ \hline
&{\textbf{FPSU.\SETUP}} (Alg. \ref{algo:FPSU:Su}):&\\
$N=(2 \delta+1)n$&&\\

   \multicolumn{3}{|c|}{\(\begin{array}{rcl}
    O^{FPSU}_{\mathcal R} \hspace*{0.5cm} \xleftarrow{\hspace{1cm}}&\ovalbox{LOKHER.\SETUP } & \xrightarrow{\hspace{1cm}}\hspace*{0.5cm} O^{FPSU}_{\mathcal S}\\
    \end{array}\)} \\

&&\\
\hline
 \multicolumn{3}{|c|}{\textbf{\NGFPSU.\QUERY}, \textbf{\NGFPSU.\RESP}
   and \textbf{\NGFPSU.\DECOD} (Alg. \ref{algo:NGFPSU:Qr},
   \ref{algo:NGFPSU:Rs}, \ref{algo:NGFPSU:Dc}):}\\
&&\\

 $\forall i \in [n], \forall j \in [d-1], r_{ij} \xleftarrow{\$}  \mathcal{K}$ &&\\
     $r_{id}=- \sum_{j=1}^{d-1}r_{ij}$&&\\
{\textbf{For all}} $j\in[d]$&&\\
 \multicolumn{3}{|c|}{\(\begin{array}{rcl}
      \cdashline{1-3}
KV_j=\{(x_{ij}+\epsilon,r_{ij})_{i\in[n],\epsilon \in[- \delta, \delta] }\}\xrightarrow{\hspace{0.5cm}}&\ovalbox{ LOKHER } & \xleftarrow{\hspace{1cm}} \{(y_{ij})_{i \in [m]} \}\hspace*{1cm}\\
    &  &\xrightarrow{\hspace{1cm}}(\hat s_{ij})_{i\in[m]} \\
      \cdashline{1-3}
    \end{array}\)} \\
    && \hspace*{1cm} $\hat t_i= \sum_{j=1}^{d} \hat s_{ij} \in {\mathcal C}$\\
&$\xleftarrow{(\hat t_i)_{i\in[m]},(\hat u_i)_{i\in[m]}}$& \hspace*{1cm}  $\hat u_i=y_i\hat t_i \in {\mathcal C}^d$\\
\hline
&\textbf{FPSU.\UNION} (Alg. \ref{algo:FPSU:Un}):&\\
 $\forall i \in [m], t_i =D_{sk}(\hat t_i)$&&\\
 if $t_i \neq 0$, $ u_i =D_{sk}(\hat u_i)$ and $X =X \cup \{ u_i/t_i  \}$ &&\\
\hline
    \end{tabular}
    \end{figure*}

\begin{algorithm}[htbp]
\caption{$\textbf{FPSU.\SETUP}(\kappa, \lambda, (Enc,Dec),n, \delta)$}\label{algo:FPSU:Su}
\begin{flushleft}
\textbf{Input:} A security parameter $\kappa$, three  sets  $\mathcal K$,  An encryption scheme $(Enc,Dec)$ over $\mathcal V$ and $\mathcal C$.\\
    \textbf{Outputs:} A set of of secret and share keys
    $O^{FPSU}_\mathcal{R}$ and $O^{FPSU}_\mathcal{S}$.
\end{flushleft}
\begin{algorithmic}[1]

\State $\mathcal{R}$: compute $N=(2 \delta +1)n$;
\State $\mathcal{R}$ and $\mathcal{S}$: compute  $(O^{FPSU}_{\mathcal R},O^{FPSU}_{\mathcal S} )\leftarrow \textbf{LOKHER.\SETUP}(\kappa, \log d+\lambda, (Enc,Dec),N)$;


\State $\mathcal{R}$: \Return $ O^{FPSU}_{\mathcal{R}}$;
\State $\mathcal{S}$: \Return $ O^{FPSU}_{\mathcal{S}}$;
\end{algorithmic}
\end{algorithm}

    \begin{algorithm}[htbp]
        \caption{\textbf{\NGFPSU.\QUERY}$(O^{FPSU}_\mathcal{S}, \textbf{Y})$}
        \label{algo:NGFPSU:Qr}
            \begin{flushleft}
            \textbf{Input:} A set of parameters generated previously $O^{FPSU}_\mathcal{S}$ and the elements of the sender set $\textbf{Y} = \{y_1, \dots, y_m\} \subset {\mathcal K}^d$ \\
            \textbf{Output:} Sequences $(\tilde I_j)_{j \in [d]}$, $(J_j)_{j \in [d]}$ and $(\beta_j)_{j \in [d]}$;
            \end{flushleft}
        \begin{algorithmic}[1]
\State $\mathcal{S}$: extract $ O^{PIR}_{\mathcal{S}}$ and $l(.)$ from $O^{FPSU}_\mathcal{S}$.
\ForAll{$j \in [d]$}
\State $\mathcal{S}$: compute    $ (\tilde I_j,J_j,\beta_j) \leftarrow \textbf{LOKHER.\QUERY}(O^{FPSU}_\mathcal{S},\{y_{1j}, \dots, y_{mj}\} )$; \Comment{$\tilde I_j$ of size $M_j \leq \alpha m$}

\EndFor

\State $\mathcal{S}$: send $(\tilde I_{j})_{j\in[d]}$ to $\mathcal{R}$;
\State $\mathcal{R}$: \Return $(J_j)_{j \in [d]}$ and $(\beta_j)_{j \in [d]}$;
\State $\mathcal{S}$: \Return $(\tilde I_{j})_{j \in[d]}$
\end{algorithmic}
    \end{algorithm}


    \begin{algorithm}[htbp]
        \caption{\textbf{\NGFPSU.\RESP}$(O^{FPSU}_\mathcal{R}, X,(\tilde I_j)_{j \in [d]})$}
        \label{algo:NGFPSU:Rs}
            \begin{flushleft}
            \textbf{Input:} A set of parameters $O^{FPSU}_\mathcal{R}$, $X$ and the sequences of encrypted indices $(\tilde I_j)_{j \in [d]}$.($I_j$ is of lenght $M$) \\
            \textbf{Output:} the right parts of the encoding vectors $(R_i)_{j \in [d]}$ and a  set of encrypted responses $(\tilde{Rs}_{j})_{ j \in [d]} $.
            \end{flushleft}
        \begin{algorithmic}[1]

\ForAll{$ i \in [n]$}
\ForAll{ $j \in [d-1]$}
\State $\mathcal{R}$: randomly generate $r_{ij} \leftarrow  \mathcal{K}_q$ \EndFor
   \State $\mathcal{R}$: compute  $r_{id}=- \sum_{j=1}^{d-1}r_{ij}$
\EndFor

\ForAll{$j \in [d]$}

\State $\mathcal{R}$: build $KV_j=\{(x_{ij}+k,r_{ij})_{i\in[n],k \in[- \delta, \delta] }\}$;
\State $\mathcal{R}$: compute \\  \hspace*{1cm}  $(R_j,(\tilde {Rs}_{ij})_{i\in[M_j]}) \leftarrow \textbf{OKHER.\RESP}(O^{FPSU}_\mathcal{R},KV_j, I_j)$
\EndFor
\State $\mathcal{R}$: send $(R_j)_{j\in[d]}$ and $(\tilde{Rs}_{j})_{j\in[d]}$  to $\mathcal{S}$;
\State $\mathcal{S}$: \Return $(R_j)_{j\in[d]}$ and $(\tilde {Rs}_{j})_{j\in[d]}$.
\end{algorithmic}
    \end{algorithm}

    \begin{algorithm}[htbp]
        \caption{\textbf{\NGFPSU.\DECOD}$(O^{FPSU}_\mathcal{S}, (R_j)_{j\in[d]}, (J_{j})_{j \in[d]},$ $ (\beta_j)_{j\in[d]}, (\tilde{Rs}_{j})_{j \in[d]},Y)$}
        \label{algo:NGFPSU:Dc}
            \begin{flushleft}
            \textbf{Input:} A set of parameters $O^{FPSU}_\mathcal{R}$, the right parts of the decoding vector $(R_j)_{j\in[d]}$, the private sequences $J$ and $\beta$ and the set of masked responses $\hat{Rs}$ and $\textbf{Y}$. \\

            \textbf{Output:} A list $(\hat t_i)_{i \in [m]} $ and  $(\hat u_i)_{i \in [m]} $.
            \end{flushleft}

        \begin{algorithmic}[1]

\ForAll{$j \in [d]$}
\State $\mathcal{S}$: compute $(\hat s_{ij})_{i \in[m]} \leftarrow \textbf{LOKHER.\DECOD}(O^{FPSU}_\mathcal{S},R_j,  J_j, \beta_j, (\tilde{Rs}_{ij})_{i \in [M_j]})$;
\EndFor

\State $\mathcal{S}$: compute  $(\hat t_i)_{i \in[m]}= (\sum_{j=1}^{d} \hat s_{ij})_{i \in[m]}$;
\State $\mathcal{S}$: compute    $(\hat u_i)_{i \in[m]}=(y_i\hat t_i)_{i \in[m]}$;

\State $\mathcal{S}$: send $(\hat t_i,\hat u_i )_{i \in[m]}$ to $\mathcal R$.
\State $\mathcal{R}$: \Return $(\hat t_i,\hat u_i )_{i \in[m]}$.
\end{algorithmic}
    \end{algorithm}

\begin{algorithm}[htbp]
\caption{$\textbf{FPSU.\UNION}(\textbf{X},(\widehat{t_{i}})_{i\in[m]},(\widehat{u_{i}})_{i\in[m]},sk)$}\label{algo:FPSU:Un}
\begin{flushleft} \textbf{Input:} A set of plaintexts
$\textbf{X}=\lbrace x_i\rbrace_{i\in[n]}\subset\mathcal{K}^d$, a set of LHE ciphertext
$(\widehat{t_{i}})_{i\in[m]}\subset{\mathcal C}$ and a set of vectors of LHE ciphertext   $(\widehat{u_{i}})_{i\in[m]}\subset {\mathcal C}^d$.\\
\textbf{Output:} A set of plaintexts $\textbf{Z}\subset{\mathcal{K}}^d$, such that $\textbf{Z}=\textbf{X}\cup\lbrace \textbf{Dec}_{sk}(\widehat{u_i})\textbf{Dec}_{sk}(\widehat{t_i})^{-1}:\textbf{Dec}_{sk}\neq 0\rbrace$.\\
\end{flushleft}
\begin{algorithmic}[1]
\State
$\mathcal{R}$: compute $\textbf{Z}\leftarrow \textbf{X}$;
\ForAll{$i \in [m]$} \State
$\mathcal{R}$: compute $t_i\leftarrow
\textbf{Dec}_{sk}(\widehat{t_i})$;
\If{$t_i\neq 0$} \State
$\mathcal{R}$: compute $\textbf{Z}\leftarrow \textbf{Z} +\lbrace
\textbf{Dec}_{sk}(\widehat{u_i})t_i^{-1}\rbrace$;
\EndIf
\EndFor
\State $\mathcal{R}$: \Return \textbf{Z};
\end{algorithmic}
\end{algorithm}

\begin{theorem}\label{thm:NGFPSU}
In the honest-but-curious adversary settings, the FPSU protocol
presented in \cref{proto:NGFPSU} and composed of the 5 algorithms
\textbf{FPSU.\SETUP},
\textbf{\NGFPSU.\QUERY},
\textbf{\NGFPSU.\RESP},
\textbf{\NGFPSU.\DECOD},
\textbf{FPSU.\UNION}
is a correct and secure fuzzy private set union scheme  with a probability of success of $1- 2^{-\lambda}$ under the LOKHER perfect correctness, obliviousness and security assumptions and a probability of success of $1- 2^{-\lambda-\log{d}}$.
\end{theorem}
The asymptotic complexity bounds are presented in \cref{tab:FPSUComplexity}.
\begin{table}[htbp]
    \centering

    \begin{tabular}{l|ccc}
\toprule
        \textbf{Procedure} & \textbf{Receiver} & \textbf{Sender} & \textbf{Comm.} \\
\midrule
        LOKHER & ${\mathcal{O}}(d\delta n \log (\delta n))$ & $\mathcal{O}(dm\log (\delta n))$ & $\mathcal{O}(dm\log (\delta n))$ \\
        Union & $\mathcal{O}(dm)$ & $\mathcal{O}(dm)$ & $\mathcal{O}(dm)$ \\
        Total & $\mathcal{O}(d\delta n \log (\delta n))$ & $\mathcal{O}(dm \log (\delta n))$ & $\mathcal{O}(dm \log (\delta n))$ \\
\bottomrule
    \end{tabular}
    \caption{Complexity of the Fuzzy Private Set Union protocol.}
    \label{tab:FPSUComplexity}
\end{table}

The proof of \cref{thm:NGFPSU} is by simulation and shown in~\cref{FPSUproof}.

Despite the seemingly significant constraint of null graphs, this
applies to several important cases.
As an example of the application of this protocol, we can cite the biometric database updates applications. There, the first element corresponds to a person's ID (not fuzzy and all different) and the other elements are dedicated to biometric scan values (fuzzy).
We are then able to construct a well-formed set of key-value pairs
while maintaining the same asymptotic complexity.
The idea is to prefix the keys by the
non-fuzzy dimension coordinate: $KV$ keys are then build from
$(x_{i1}||x_{il}+ \epsilon)$, while queries are build from
$(y_{i1}||y_{ij})$.
This generates a new protocol, denoted \NFFPSU, that uses the
\NGFPSU protocol on the axis range $[2..d]$.

In the following subsections, we work on relaxing this null graph
strong assumption.
Mostly the ideas will be to modify the keys of the key-value store, so
that it is still possible to obliviously detect when $y$ is in the
union of balls. We propose several conditions on the structure of the
induced graph, and adaptations, guaranteeing that the protocol remains
efficient.

\subsection{FPSU on any graphs (with a small chromatic number)}\label{LAYFPSU}

To be able to apply our protocol on any graph, we divide the balls
into layers, and then apply~\cref{NGSection} layer by layer.
For this, we build a partition of the set $X=\{x_1,x_2, \ldots, x_n\}$
into $\chi$ subsets such that $X = \sqcup_{l=1}^{\chi} Z_l$ and for
all $l$, the  subgraph  $G_{Z_l}$ of $G_X$  is a null graph.
The main steps of the protocol are now as follows:
\begin{itemize}
\item
The prover builds $\chi$ layers $(Z_l)_{l \in [\chi]}$ where $X =
\sqcup_{l=1}^{\chi} Z_l$ and the subgraph $G_{Z_l}$ of $G_X$ is a null
graph. This partition can be perfomed by a coloring algorithm, such as
DSATUR, as shown in~\cref{ssec:chromatic}.
\item For each layer $l$, the prover defines the distinct keys
  $l||x_{ij}+\epsilon$, for all $x_i \in Z_l$, for all $j \in[d]$ and for all
  $\epsilon \in[-\delta,\delta]$.
$KV_j^l$ represents a subset the key-value pairs for for all $x_i \in
Z_l$ and $\epsilon \in[-\delta,\delta]$, and $KV_j^l$ is well-formed
thanks to the colored partition.
\item For all $j \in[d]$, the prover then gathers
  $KV_j=\cup_{l=1}^{\chi}KV_j^l$. $KV_j$ is also well-formed since
  $KV_j^l$ is well-formed and as all keys start with the layer index
  $\l$, no conflict can occur with any key from another $KV_j^{l'}$
  with $l' \neq l$.
\item For each point $y_i$ of her dataset $Y$, the sender builds the $\chi$ keys $ (l||y_{ij})_{j \in [d], l \in [\chi]}$ and receives the corresponding encrypted values $\hat s_{ijl}$.
\item The sender computes
  $\hat{t_{i}}=\prod_{l\in[\chi]}\sum_{j\in[d]}\hat{s_{ ijl}}$.
    We now have that if $y_i$ belongs to
    $\cup_{i\in[n]}{\mathcal{B}}_\delta(x_i)$ then
    $\hat{t_{i}}=\hat{0}$.
    Reciprocally, with high probability, $\hat{t_{i}}=\hat{0}$ implies
    that $y_i$ belongs to $\cup_{i\in[n]}{\mathcal{B}}_\delta(x_i)$.
    The computation of $t_i$ requires that the encryption scheme is
    fully homomorphic and thus, FOKHER has now to be used as the
    building block.
\item $\textbf{FPSU.\UNION}$ remains unchanged,
    with the technique of~\cite{DBLP:conf/acns/Frikken07}
    to retreive the elements actually in the union:
    upon decryption of $\hat{t_i}$ and
    $\hat{u_i}=y_i{\ltimes}\hat{t_i}$,
    either $t_i=0$, or $y_i=u_it_i^{-1}$.
\end{itemize}

The protocol
\textbf{\LAYFPSU} uses $\textbf{Lay.KeyValue.Build}$
(Algorithm~\ref{algo:LKV}) as a subroutine.
The new algorithms $\textbf{\LAYFPSU.\QUERY}$,
$\textbf{\LAYFPSU.\RESP}$, as well as  $\textbf{\LAYFPSU.\DECOD}$, are variants
of their counter parts in~\cref{proto:NGFPSU}.

 \begin{algorithm}[htbp]
        \caption{\textbf{Lay.KeyValue.Build}$( \textbf{X}, \delta)$}
        \label{algo:LKV}
            \begin{flushleft}
            \textbf{Input:} $\textbf{X}=\{x_1, \ldots,x_n\}\subset {\mathcal K}^d$,and  $\delta$.\\
            \textbf{Output:} The number of layers $\chi$ and a set of key-value pairs $KV_j$ for each dimension $j \in[d]$.
            \end{flushleft}
        \begin{algorithmic}[1]

\State $\mathcal{R}$: initialize $KV_j \leftarrow \emptyset $ for all $j\in[d]$;
\State $\mathcal{R}$: compute $(Z_l)_{l\in [\chi]}=DSATUR(X, \delta)$;
\ForAll{$l \in [\chi]$}
\ForAll{$i \in [lenght(Z_l)]$}
\State $\mathcal{R}$: $r_{id} \xleftarrow{} 0$;
\ForAll{$j \in [d]$}
\State $\mathcal{R}$: $r_{ij} \xleftarrow{\$} {\mathcal V}$;
\State $\mathcal{R}$: $KV_j \xleftarrow{}KV_j,(l||x_{ij}+\epsilon,r_{ij})_{\epsilon \in[- \delta,\delta]}$;
\State $\mathcal{R}$: $r_{id} \xleftarrow{} r_{id}-r_{ij}$;
\EndFor
\State $\mathcal{R}$: $KV_d \xleftarrow{}KV_d,(l||x_{id}+\epsilon,r_{id})_{\epsilon \in[- \delta,\delta]}$
\EndFor
\EndFor
\State $\mathcal{S}$: \Return $( \{KV \}_{j})_{j\in[d]}$.
\end{algorithmic}
    \end{algorithm}

\begin{theorem}\label{thLKVB}
The algorithm \textbf{Lay.KeyValue.Build} described in Algorithm
\ref{algo:LKV} builds $d$ well-formed sets of key-value pairs
$\{(k_{ij},v_{ij})_{i \in [N]} \}$ with $N=(2 \delta+1) n$.

\end{theorem}

Thanks to \cref{thLKVB}, we obtain the following correctness and security proof.
\begin{theorem}\label{thm:LAYFPSU}
In the honest-but-curious adversary settings, the FPSU protocol
and composed of the 5 algorithms
\textbf{FPSU.\SETUP},
\textbf{\LAYFPSU.\QUERY},
\textbf{\LAYFPSU.\RESP},
\textbf{\LAYFPSU.\DECOD},\\
\textbf{FPSU.\UNION}
is a correct and secure fuzzy private set union scheme  with a probability of success of $1- 2^{-\lambda}$ under the OKHER perfect correctness, obliviousness and security assumptions and a probability of success of $1- 2^{-\lambda-\log{d}}$.
\end{theorem}

 The proof is similar to that of Theorem \ref{thm:NGFPSU}.
The number of layers $\chi$ corresponds to the number of colors
obtained during the DSATUR algorithm. This $\chi$ is an upper bound on
the actual chromatic number $\chi(G_X)$ of the graph
$G_X$ viewed as a simple graph (without weight on the edges).
Then a classical upper bound gives $ \chi(G_X) \leq deg(G_X).$
There are two main weaknesses in this approach:
(1) (a bound on) the actual value $\chi$ has to be sent to the sender;
(2) the chromatic number can be large.
If we consider a random graph $G$ of $n$ vertices,
$\chi(G)\sim{O(n/\log{n})}$ (see
\cite{bollobas1998random,ABE2025100600} for instance).

In the next section, we thus propose other approaches that better
control the computational and communication complexities bounds, while
still being applicable to large families of data sets.


\subsection{FPSU on a-exclusive graph}
\label{ExcSection}

In this section, we propose a new construction of the set of key-value pairs which fits for any set of points $\textbf{X} \subset {\mathcal K}^d$ and any  radius $\delta$
such that: there exists a dimension $a \in[d]$ for which $G_X$ is
a-exclusive.
In other words:
if $\exists (i,j)\in[n]^2, i \neq j$,\,s.t.
$[x_{ia}-\delta,x_{ia}+\delta]\cap[x_{ja}-\delta,x_{ja}+\delta]\neq\emptyset$
then
$\forall\epsilon\in[1..d]{\setminus}{a},[x_{i\epsilon}-\delta,x_{i\epsilon}+\delta]\cap[x_{j\epsilon}-\delta,x_{j\epsilon}+\delta]=\emptyset$.
The construction of the sets of key-value pairs then ensures that the
keys are all different, per dimension:
\begin{itemize}
\item We do not consider the structure of the balls in the $a$-axis:
  to build the well-formed set of key-value pairs for the other
  axes. We then add each value $v$ of
  $I=\cup_{i=1}^n[x_{ia}-\delta,x_{ia}+\delta]$ as a prefix to the keys.
\item For all $v \in I$,
we consider the set of indices $L_v=\{i \in[n],v \in[x_{ia}- \delta, x_{ia}+ \delta] \}$ of balls which projection into the first axis hits $v\in{I}$.
\item For all different indices $( i,i') \in L_v^2$, for all the axis
  $j\in [d] {\setminus} \{a \}$ we have the property that
  $[x_{ij}-\epsilon,x_{ij}+\epsilon]\cap[x_{i'j}-\epsilon,x_{i'j}+\epsilon]=\emptyset$;
  since the graph is $a$-exclusive.
  Then, we sample random values $r_{ijv}$ for all
  $j\in[d-1]{\setminus}\{a \}$ and compute
  $r_{idv}=-\sum_{j\in[d-1]{\setminus}\{a\}} r_{ijv}$, in order to
  build the well-formed:
$$KV_j=\{(a||v||x_{ij}+\epsilon,r_{vij})_{v\in{I},i\in{L_v},j\in[2..d]}\}$$
The size of the whole key-value store becomes $(2 \delta+1)^2nd$.
(Note that if $d=a$, the range for $j $ becomes $j \in [d-2]$, and we
build $r_{i(d-1)v}$ as the negative sum of the $r_{ijv}$.)
\item  By construction, we thus have that
  $y \in \cup_{i\in[n]}{\mathcal B}_\delta(x_i)$ must satisfy $\sum_{j \in [d] {\setminus} \{a \}}  s_j=0$ where $s_j$ is the value corresponding to the key $y_{j}$ in $KV_j$.
\item In order for $a$ to remain private, the client builds $d$
  different requests $(a||y_{ia}||y_{ij})$ for each $y_{ij}$.
\end{itemize}
Algorithm \ref{algo:ExcKV} gives the formal constructions of the $d$
sets of key-value pairs.

  \begin{algorithm}[htbp]
        \caption{\textbf{Exc.KeyValue.Build}$( \textbf{X}, \delta,a)$}
        \label{algo:ExcKV}
            \begin{flushleft}
            \textbf{Input:} $\textbf{X}=\{x_1, \ldots,x_n\}$, $\delta$ and $a \in[1..d]$ such that $G_X$ is $a$-exclusive.\\
            \textbf{Output:} A set of key-value pairs $KV_j$ for each dimension $j \in[1..d]{\setminus} \{a \} $.
            \end{flushleft}
        \begin{algorithmic}[1]
\State $\mathcal{R}$: build $I= \cup_{i=1}^n [x_{ia}- \delta, x_{ia}+ \delta]$;
\ForAll{$v \in I$}
\State $\mathcal{R}$:  build $L_v=\{i \in[n],v \in[x_{ia}- \delta, x_{ia}+ \delta] \}$;

\ForAll{$i \in L_v$}
\State $\mathcal{R}$: $r_{id} =0$
\State $\mathcal{R}$: if $a=d$ then $d'=d-1$ else $d'=d$
\ForAll{$j \in [d'-1] {\setminus} \{a \}$}
\State $\mathcal{R}$: $r_{ij} \xleftarrow{} {\mathcal V}$;
\State $\mathcal{R}$: $KV_j \xleftarrow{}KV_j, ( a||v||x_{ij}+\epsilon,r_{ij})_{\epsilon \in[- \delta,\delta]}$;
\State $\mathcal{R}$: $r_{id'} \xleftarrow{} r_{id'}-r_{ij}$;
\EndFor

\State $\mathcal{R}$: $KV_{d'} \xleftarrow{}KV_{d'},(a||v||x_{id'}+\epsilon,r_{id'})_{ \epsilon \in[- \delta,\delta]}$
\EndFor
\EndFor
\State $\mathcal{R}$: build a well-formed set $KV_a$ with$(2 \delta +1)^2 \times n$ randoms of $ \mathcal K \times \mathcal V$.
\State $\mathcal{R}$: \Return $( \{KV \}_{j})_{j\in[d]}$.
\end{algorithmic}
    \end{algorithm}


\begin{theorem}
The algorithm \textbf{Exc.KeyValue.Build} described Algorithm \ref{algo:ExcKV} builds $d$ well-formed sets of key-value pairs $\{ (k_{ij},v_{ij})_{i \in [N]} \}$ with $N=(2 \delta+1)^2 n$.
\end{theorem}

The resulting \textbf{\EXCFPSU} protocol for $a$-exclusive induced
dataset graph
uses as subroutines the
algorithm \ref{algo:ExcKV} to build the sets of key-value pairs $KV$
as well as the \textbf{FOKHER} subroutine. The  algorithms
\textbf{\EXCFPSU.\QUERY}, \textbf{\EXCFPSU.\RESP},
\textbf{\EXCFPSU.\DECOD} are variants of algorithms \textbf{\NGFPSU.\QUERY}, \textbf{\NGFPSU.\RESP},
\textbf{\NGFPSU.\DECOD}.

\begin{theorem}\label{thm:EXCFPSU}
In the honest-but-curious adversary settings, the FPSU protocol
composed of the 5 algorithms
\textbf{FPSU.\SETUP},
\textbf{\EXCFPSU.\QUERY},
\textbf{\EXCFPSU.\RESP},
\textbf{\EXCFPSU.\DECOD} and
\textbf{FPSU.\UNION}
 is a correct and secure fuzzy private set union scheme  with a
 probability of success of $1- 2^{-\lambda}$ under the \textbf{FOKHER}
 perfect correctness, obliviousness and security assumptions and a
 probability of success of $1- 2^{-\lambda-\log{d}}$.
\end{theorem}

\subsection{For d-stripable dataset }\label{sec:STRFPSU}

Here, we use the approach proposed in \cref{ExcSection} on distinct
strips of the balls: if the set $X=\{x_1,x_2, \ldots,x_n\}$ can be partitioned
into $d$ subsets, such that $X = \sqcup_{a=1}^d Z_a$ and all the
subgraphs $G_{Z_a}$ of $G_X$  are $a$-exclusive, then we obtain an
efficient protocol.
\begin{definition}\label{defStrip}
A set $X=\{x_1,x_2, \ldots, x_n\}$ is $d$-stripable if there exists a
partition  $X = \sqcup_{a=1}^d Z_a$ such that all  the  subgraphs
$G_{Z_a}$ of $G_X$  are a-exclusive.
\end{definition}

This variant of \EXCFPSU is composed of the same four algorithms \textbf{FPSU.\SETUP},
\textbf{\EXCFPSU.\QUERY},
\textbf{\EXCFPSU.\DECOD},\\
\textbf{FPSU.\UNION}.
The algorithm \textbf{\EXCFPSU.\QUERY} is modified to accept
$d$-stripable sets and given as  \textbf{\STRFPSU.\RESP}.
 More precisely, at the beginning of
\textbf{\STRFPSU.\RESP}, the receiver builds the $d$ strips $Z_a$, of
the $d$-stripable $X$, such that $X = \sqcup_{a=1}^d Z_a$ and
$G_{Z_a}$ is $a$-exclusive.
For each strip, she builds the key-value pairs using
$\textbf{Exc.KeyValue.Build}(Z_a,\delta,a)$ and obtains, for the strip
$a$, the sets of key-value pairs $KV^a$. Then, she can just run
\textbf{LOKHER.\RESP}.

\begin{theorem}\label{thm:STRFPSU}
In the honest-but-curious adversary settings, the FPSU protocol
composed of the 5 algorithms
\textbf{FPSU.\SETUP},
\textbf{\EXCFPSU.\QUERY},
\textbf{\STRFPSU.\RESP},
\textbf{\EXCFPSU.\DECOD},\\
\textbf{FPSU.\UNION}
is a correct and secure fuzzy private set union scheme   with a
probability of success of $1- 2^{-\lambda}$ under the \textbf{FOKHER}
perfect correctness, obliviousness and security assumptions and a
probability of success of $1- 2^{-\lambda-\log{d}}$.
\end{theorem}
\begin{proof}
Once the d strips have been designed, the proof is similar to the
proof of the \EXCFPSU~protocol. 
\end{proof}

There remains to generates a partition
$X=\sqcup_{a \in [d]}Z_a$ such that, for all $a \in [d]$, $G_{Z_a}$ is
$a$-exclusive. We give an idea of how to build such a partition in the
following examples:
\begin{itemize}
\item If $G_X$ is a-exclusive, use \textbf{\EXCFPSU};
\item If the chromatic number of $G_X$ is not greater than $ d$, use
 \textbf{ \LAYFPSU};
\item If all the vertices of $G_X$ are a-exclusive for some $a\in[d]$,
  then the $Z_a$ contain those vertices;
\item If $d=2$ and $G^{1,2}_X=(V_X,E_X^{1,2})$ is a bipartite graph
  (i.e. has no odd cycles), then $Z_1$ and $Z_2$ are the two
  independent set of vertices;
\item If the degree of $G_X$ is at most $\frac{3}{2} d -1$, then the
  constructive proof of~\cref{thm:STRP} in~\cref{app:dstrip:thm} allows to
  build the partition.
\end{itemize}
Furthermore, in Appendix \ref{app:dstrip:thm}, we
propose the heuristic of~\cref{algo:Str}, which outputs a partition
$X=\sqcup_{a=1}^d{Z_a}$, but may fail.
Characterizing all sets $X$ such that~\cref{algo:Str} outputs a
partition remains an open problem.

{\small
\bibliographystyle{plainurl}
\bibliography{bibfile}
}
\appendix
\crefalias{section}{appendix}
\section{\textbf{OKHER} correctness and soundness}\label{OKHERproof}

In this appendix, we prove \cref{thm:OKHER} which states that
given  correct and secure PIR and SOKVS schemes, the protocol defined
in~\cref{proto:OKHER} is correct and secure under the honest-but-curious
model.
%
%
We first build the views of the receiver and the sender.
The view of the receiver is composed of its input $KV= \{(k_i,v_i)_{i
  \in [N]}\}$, the keys of the encryption scheme, the set  of
key-encrypted value pairs $KC=\{(k_i,Enc_{pk}(v_i))_{i \in [N]}\}$,
its SOKVS and PIR views, $\textbf{view}_{\mathcal R}^{SOKVS}(KC, \bot)$
as well as
$\textbf{view}_{\mathcal R}^{PIR}(L, NZIndices(l(q_i))_{i\in [m]})$,
where $q_i$ is a sender's query (and $L$ comes from
$\textbf{view}_{\mathcal R}^{SOKVS}(KC, \bot)$):

$
\textbf{view}_{\mathcal{R}}^{OKHER}(KV,Q)=\left(KV,pk,sk,KC,\textbf{view}_{\mathcal{R}}^{SOKVS}(KC,\bot),\textbf{view}_{\mathcal{R}}^{PIR}(L,NZIndices(l(q_i))_{i\in[m]})\right)$
Now, the view of the sender is composed of its input $Q= \{(q_i)_{i \in [m]} \}$,
the public key of the encryption scheme $pk$, and its PIR and SOKVS views:

$\textbf{view}_{\mathcal{S}}^{OKHER}(KV,Q)=\left(Q,pk,\textbf{view}_{\mathcal{S}}^{SOKVS}(KC,\bot), \textbf{view}_{\mathcal{S}}^{PIR}(L,NZIndices(l(q_i))_{i\in[m]})\right)$

The correctness proof has two cases. (1) If
$q_i\in\{k_1,\ldots,k_N\}$,
from the correctness of the PIR protocol,
the sender receives the correct $\alpha$ non-zero indices of $l(q_i)$;
from the correctness of the SOKVS  we have that
$\textbf{OKHER.\DECOD}$ outputs $e_i= Enc_{pk}(v_j)$ where $j$
satisfies  $q_i=k_j$, w.h.p.
(2) If $q_i \notin \{k_1,\ldots,k_N\}$, as the encryption scheme is
IND-CPA and SOKVS is correct, the output $e_i$ is pseudo-random.

For the security, we assume that the Batch-PIR  protocol is  secure
and that the SOKVS is doubly oblivious (since the encryption scheme is
IND-CPA).
More precisely, for all vectors $L \subset \mathbb Z  $ and for all
$I\subset [1..|L|] $, we have that
$\textbf{view}^{PIR}_{\mathcal{S}}(L,I)\equiv_c\textbf{Sim}^{PIR}_{\mathcal{S}}(I,L[I])$
and
$\textbf{view}^{PIR}_{\mathcal{R}}(L,I)\equiv_c\textbf{Sim}^{PIR}_{\mathcal{R}}(L,\bot)$
We define simulators $Sim^{OKHER}_{\mathcal{R}}$ and
$Sim^{OKHER}_{\mathcal{S}}$ for the corrupt receiver (resp. sender),
and we prove that they are indistinguishable from the
corresponding views.
\paragraph{Corrupted receiver}
$\textbf{sim}_{\mathcal{R}}^{OKHER}(KV,\bot \}$ simulates the view of
the corrupt receiver as follows:
(1) It generates the pair of keys $(pk,sk)$ of the encryption scheme
  $(Enc,Dec)$;
(2) From $KV$ and the key $pk$, it computes
  $KC=\{(k_i,Enc_{pk}(v_i))\}$ and uses the {\bf{SOKVS}} algorithms
to obtain the mappings $l(.)$ and $r(.)$ and
the vectors $L$, $R$;
(3) It randomly generates a sequence of queries $Q$ and computes
  $I_i=NZIndices(l(q_i))$ for all $i \in[m]$;
(4) It invokes the PIR receiver simulator
  $\textbf{sim}_{\mathcal{R}}^{PIR}(L,\bot)$.

Therefore,
$\textbf{sim}_{\mathcal{R}}^{OKHER}(KV,\bot)=\left(KV,pk,sk,KC,\textbf{view}_{\mathcal{R}}^{SOKVS}(KC,\bot),\right.$
$\left.\textbf{Sim}_{\mathcal{R}}^{PIR}(L,\bot)\right)$.
Then, since $\textbf{view}_{\mathcal{R}}^{PIR}(L,I_i) \equiv_c \textbf{sim}_{\mathcal{R}}^{PIR}(L,\bot)$
we have that
$ \textbf{sim}_{\mathcal{R}}^{OKHER}(KV,\bot) \equiv_c \textbf{view}_{\mathcal{R}}^{OKHER}(KV,Q) $.
\paragraph{Corrupted sender}
$\textbf{sim}_{\mathcal{S}}^{OKHER}(Q,(e_i)_{i \in [m]})$ simulates
the view of the corrupt sender as follows:
(1) To construct $L$ and $R$, from the setup, it knows the mappings   $l: {\mathcal K} \rightarrow \{0,1\}^{|L|} $ of weight $\alpha$  and  $ r: {\mathcal K} \rightarrow \{0,1\}^{|R|}$  such that $|L| \sim N$ and $|R|\sim \log N$.
It checks that the $m$ values  $(e_i)_{i \in [m]} \subset {\mathcal
  C}$ are distinct.
(Otherwise, only the $m'$ different values of $e$ are considered).
It randomly generates $(f_i)_{i \in [N-m]} \subset {\mathcal C}$ and
shuffles $e$ and $f$ into a third sequence $(g_i)_{i \in [N]}$. In the
same way, it randomly chooses $N-m$ keys and shuffles them, with the
same permutation, with the sequence $(q_i)_{i \in [m]}$ in order to obtain a
sequence $(q'_i)_{i \in[N]}$. Then it computes  two vectors $L$ and
$R$ satisfying $\forall i \in[N], <l(q'_i)||r(q'_i),L||R>=g_i$.
By construction of $l(.)$ and $r(.)$, $L$ an $R$ do exist and the
matrix $\left(l(q'_i)||r(q'_i)\right)$ is full rank w.h.p.~\cite{298226}.
From the double obliviousness of SOKVS, $L||R$ is indistinguishable
from a random vector;

(2) It computes, for all $i \in[m]$,
  $\textbf{PIR.\QUERY}(O_{\mathcal S}^{PIR}, I_i)$
such that $I_i=NZIndices(l(q_i))$ and then extracts the needed values
$L^{I_i}$ from~$L$;
(3) By construction, for all $i \in[m]$, $e_i=<r(q_i),R>+\sum_{j=1}^{\alpha} L_j^{I_i}$.

Finally, the view simulated by $\textbf{sim}_{\mathcal{S}}^{OKHER}$ is
computationally indistinguishable from the real one and we have:
$$\textbf{view}^{OKHER}_{\mathcal S}(KV,Q)\equiv_c \textbf{sim}_{\mathcal{S}}^{OKHER}(Q, output_{\mathcal S}(KV,Q))$$
\section{\textbf{FPSU} correctness and soundness}\label{FPSUproof}

In this section, we prove \cref{thm:NGFPSU} which states that
 the FPSU protocol presented in \cref{proto:NGFPSU} is a correct and secure fuzzy private set union scheme  under the OKHER perfect correctness, obliviousness and security assumptions.

As previously we consider  honest-but-curious adversary settings and we build the views of the receiver and the sender.
(1) The view of the receiver is composed of:
\begin{itemize}
\item its input $X= \{(x_i)_{i \in [n]}\}$ and the radius $\delta$,
\item the sequence of key-value pairs for all dimension $j\in [d]$,
 $KV_j=\{(x_{ij}+\epsilon,r_{ij})_{i\in[N],\epsilon\in[-\delta.. \delta]}\}$ where $(r_{ij})_{i\in[N],j\in[d-1]}$
are random values under  $ {\mathcal V}$ and
$r_{id}=-\sum_{j=0}^{d-1}r_{ij}$;
\item the view $\textbf{view}_{\mathcal R}^{OKHER}(KV, Y)$ defined in
  \cref{OKHERproof};
\item the sequences $(\hat t_i)_{i \in[m]}$, $(\hat u_i)_{i \in[m]}$  sent by the sender:

$\textbf{view}^{FPSU}_{\mathcal{R}}(X,Y)=\left(X,\delta,(r_{ij})_{i\in[n],j\in[d]},
\textbf{view}^{OKHER}_{\mathcal{R}}(KV_j,Y_j)_{j\in[d]},
(\hat{t_i})_{i\in[m]},(\hat{u_i})_{i\in[m]}\right)$
\end{itemize}
(2) The view of the sender is composed of:
\begin{itemize}
\item its input $Y= \{(y_i)_{i \in [m]} \}$;
\item the view $\textbf{view}_{\mathcal S}^{OKHER}(KV, Y)$;
\item $(\hat t_i)_{i \in[m]}$, and $(\hat u_i)_{i \in[m]}$.

$\textbf{view}^{FPSU}_{\mathcal{S}}(X,Y)=\left(Y,\textbf{view}^{OKHER}_{\mathcal{S}}(KV_j,Y_j)_{j\in[d]},
(\hat{t_i})_{i\in[m]},(\hat{u_i})_{i\in[m]}\right)$
\end{itemize}

Note that $(\textbf{view}^{OKHER}_{\mathcal{R}}(KV_j,Y_j))_{j \in[d]}$
share a single setup $O^{FPSU}_{\mathcal{R}} $ for all
$j\in[d]$. Similarly,
$(\textbf{view}^{OKHER}_{\mathcal{S}}(KV_j,Y_j))_{j\in [d]}$ share a
single setup $O^{FPSU}_{\mathcal{S}}$ for all $j\in [d]$.
The \textbf{OKHER} protocol has a probability of success of $(1-
2^{-\log d-\lambda}$ so the $d$ intricate $OKHER$ protocols have a
probability of success of
$(1-\frac{2^{-\lambda}}{d})^d\geq{(1-2^{-\lambda})}$.
The correctness of the scheme relies on the fact that an element $y$
owned by the sender has to be added to the receiver's set $X$ if, and
only if $y \notin \cup_{i \in[n]} {\mathcal B}_\delta(x_i)$.
For all the $j \in[d]$ dimensions, the receiver builds a key-value
sequence $(k_{ij},\hat r_{ij})_{i \in{(2 \delta+1)n}, j\in[d]}$  such
that
$\forall{y}\in\cup_{i\in[n]}{\mathcal{B}}_\delta(x_i),\forall\in[d],\exists{l_j}\in[(2\delta+1)n],y_j=k_{l_jj}\wedge\sum_{j=1}^d\hat{r_{l_jj}}=\hat{0}$.
With the axis disjoint assumption,
$(k_{ij})_{i\in{(2\delta+1)n},j\in[d]}=\sqcup_{i\in{n}}[x_{ij}-\delta,x_{ij}+\delta]$
is a disjoint union for all dimensions.
Therefore, there is no conflict during the generation of the randoms
and the construction of $KV$ is straightforward.
Then, we consider two cases:\\
(1) If $y \in \cup_{i \in[n]} {\mathcal{ B}}_\delta(x_i)$,
since we have axis disjoint balls, there exists a unique $i$ such that
$\forall j \in [d]$, $y_j \in [x_{ij} -\delta, x_{ij}+\delta]$. So
$\forall j \in[d], y_j$ belongs to the set of keys of $KV_j$ and
$\textbf{OKHER.\DECOD}$ outputs $(\hat s_{j})_{j\in [d]}$ where
$\hat{s_{j}}$ are the encryptions of values $s_{j}$  such that
$t=\sum_{j\in[d]}s_{j}=0$, by construction. The sender builds
$\hat{t}=\sum_{j\in [d]}  \hat s_{j}$ and $\hat u=y \hat t$ and sends
$(\hat{t},\hat{u})$ to the receiver. The receiver deciphers
$(\hat{t},\hat{u})$ and obtains $(0,0)$.\\
(2) If $b \notin \cup_{i \in[n]} {\mathcal{ B}}_\delta(x_i)$, no
$i\in[n]$ is such that for all
$j\in[d],b_j\in[x_{ij}-\delta,x_{ij}+\delta]$.
Thus, $\textbf{OKHER.\DECOD}$  outputs  $(\hat s_{j})_{j\in [d]}$ such
that $s_j=r_{kj}$,  if it exists  $k$ such that $b_j \in
[x_{kj}-\delta, x_{kj}+\delta]$ or a random value otherwise.
Therefore, w.h.p., these $d$ values $b_j$ can not be associated to $d$
keys whose sum is $0$ and $t=\sum_{j\in [d]} s_j \neq 0$.

For the security, we have that for all key-pair set $W \subset
\mathcal K \times \mathcal V $ and for all $Z \subset \mathcal K$,
$\textbf{view}^{OKHER}_{\mathcal{S}}(W,Z)\equiv_c\textbf{Sim}^{OKHER}_{\mathcal{S}}(Z,output_S(W,Z))$
and
$\textbf{view}^{OKHER}_{\mathcal{R}}(W,Z)\equiv_c\textbf{Sim}^{OKHER}_{\mathcal{R}}(W,\bot)$.
We thus consider simulators $Sim^{FPSU}_{\mathcal R}$ and
$Sim^{FPSU}_{\mathcal S}$ for
the corrupt receiver and sender.
\paragraph{Corrupted sender}
$\textbf{sim}_{\mathcal{S}}^{FPSU}(\textbf{Y},\bot)$ simulates the
view of the corrupt sender as follows:
(1) $ \textbf{sim}_{\mathcal{S}}^{FPSU}$ uniformly samples
  $(s_{ij}) \subset {\mathcal V}$  and computes
  for all $i$ and $j$, $\hat s_{ij}=Enc_{pk}(s_{ij})$;
(2) it invokes $d$ times the OKHER's sender simulator with only one
  setup, that is
  $\textbf{sim}_{\mathcal{S}}^{OKHER}(Y_j,(\hat{s_{ij}})_{i\in[m]})$,
  with a probability of success of
  $(1-\epsilon)^d$. $\textbf{sim}_{\mathcal{S}}^{FPSU}$ adds the
  output to
  the view.
(3) Eventually, it computes for all $i \in[m]$, $\hat t_i=
  \sum_{j=1}^{d} \hat s_{ij}$ and  $\hat u_i=y_i\hat t_i$ and adds them
  to its view.

The view simulated by $\textbf{sim}_{\mathcal{S}}^{FPSU}$ is
computationally indistinguishable from the real one as, for all set of
key-pair values $KV$, we have:
$$\textbf{sim}_{\mathcal{S}}^{OKHER}(Y,\hat{s_{ij}})\equiv_c\textbf{view}_{\mathcal{S}}^{OKHER}(KV,Y)$$
And:
\begin{equation*}\begin{split}
\textbf{sim}_{\mathcal{S}}^{FPSU}(Y, \bot)
& \equiv
\left ( (\textbf{sim}^{OKHER}_{\mathcal{S}}(Y, (\hat s_{ij})_{i \in
    [m]}))_{j \in[d]}, (\hat t_i,\hat u_i)_{i \in[m]}
\right )
\\
& \equiv_c \left ( (\textbf{view}^{OKHER}_{\mathcal{S}}(KV_j,Y_j ))_{j
    \in[d]}, (\hat t_i,\hat u_i)_{i \in[m]} \right ) \\
& \equiv \textbf{view}_{\mathcal{S}}^{FPSU}(\textbf{X},\textbf{Y})
\end{split}
\end{equation*}
\paragraph{Corrupted receiver}
$\textbf{sim}_{\mathcal{R}}^{FPSU}(\textbf{X},\textbf Z )$, such that
$\textbf Z=\textbf{X}\sqcup \{ y_i\in \textbf{Y}, \forall k\in[n]
y_{i} \notin {\mathcal B}_\delta(x_k) \}$, simulates the view of the
corrupt receiver as follows:
(1) From $Z$, it takes the  $|Z| -n $ elements  not in $\cup_{i \in [n]} {\mathcal B}_{\delta}(x_i)$ and randomly chooses $m+n-|Z|$ elements in $\cup_{i \in [n]} {\mathcal B}_{\delta}(x_i)$ to build $Y^*$;
(2) It randomly chooses $(r_{ij})_{i \in[n], j \in[d-1]}$  and builds
  $r_{id}= - \sum_{j=1}^{d-1} r_{ij}$;
(3) For all dimension $j$, it builds the key-value pair sequences
  $KV_j=\{(x_{ij}+\epsilon,r_{ij})_{i\in[n],\epsilon\in[-\delta,\delta]} \} $;
(4) it invokes the OKHER receiver simulator $\textbf{sim}_{\mathcal{R}}^{OKHER}(KV_j,\bot)$ for $j\in[d]$ and adds the intermediate output to its view;
(5) it generates $( t^*_k)_{k \in[m]}$, such that $t^*_k=0$ if
  $Y^*_k \in \cup_{i \in [n]} {\mathcal B}_{\delta}(x_i)$ and $t_k$ is
  randomly sampled from ${\mathcal V}^d$ otherwise; it encrypts
  $(t^*_k)_{k \in[m]}$ and build
  $(\hat{u^*_i})_{i\in[m]}=(y^*_i\hat{t^*_i})_{i\in[m]}$. As the LHE
  is IND-CPA, $\hat t^*_i$ is indistinguishable from $\sum_{j=1}^d\hat{s_{ij}}$.

Any subset of $\cup_{i \in [n]} {\mathcal B}_{\delta}(x_i)$ of size  $m+n-|Z|$  can lead to an
indistinguishable simulation. Then, as the sender's input set has been properly
simulated,
\begin{equation*}\begin{split}
\textbf{sim}_{\mathcal{R}}^{FPSU}(\textbf{X},\textbf{X} \cup
\textbf{Y})
&\equiv_c \left (X, \delta,r_{ij},\textbf{sim}^{OKHER}_{\mathcal{R}}(KV_j,\bot),\hat{t_i},\hat{u_i}\right)_{i \in[n],j \in[d]}\\
&\equiv_c \left (X, \delta, r_{ij}\textbf{view}^{OKHER}_{\mathcal{R}}(KV_j,Y_j),\hat{t_i},\hat{u_i}\right)_{i \in[n],j \in[d]}\\
&\equiv_c \textbf{view}_{\mathcal{R}}^{FPSU}(\textbf{X},\textbf{Y})
\end{split}
\end{equation*}


\section{Construction of d-stripable graphs}\label{app:dstrip:thm}
\begin{theorem}\label{thm:STRP}
Let $X= \{x_1, \cdots, x_n\}$ be a set of points and $\delta$ a threshold
such that the degree of the induced labeled graph $G_X$ is at most
$(\frac{3}{2}d -1)$, then $X$ is d-stripable.
\end{theorem}
\begin{proof}
The proof is constructive. We first consider $d$ empty sets $Z_a$ $(a \in[d])$. We want to assign all elements of X to one subset $Z_a$ such that $X= \sqcup_{a \in [d]} Z_a$ and the subgraphs $G_{Z_a}$ are a-exclusive.
For a given element $v \in V_X$,
(1) If  $\exists{a}\in [d]$ such that  $v$ has no a-edge
($\forall{i}\in{V{\setminus}\{v\}}, a \notin\phi((v,i))$) then assign $x_v$ to $Z_a$ and remove $x_v$ to $X$.
(2) else, if there exists a label $a$ such that the subgraph
  $G_{\sigma(v)}$, with vertex set $v$ and its neighbors, is
  a-exclusive, i.e.,
  $\forall{i}\in{V},a\in\phi((v,i))\implies\forall{b}\in[d]{\setminus}\{a\},b\notin\phi((v,i)))$, then assign $x_v$ to $Z_a$ and remove $x_v$ from $X$.
(3) else, $v$ is the source of a set of edges such that all the $d$
  labels are present and $G_{\sigma(v)}$ is not a-exclusive for any $a
  \in[d]$.
Therefore, there must exist some edges $(v,v')$ with at least two
labels. Thus $\sigma(v) {\setminus} \{v\}$ has at most
$\frac{d}{2}+(deg(G_X)-d)=\frac{3}{2}d-1-\frac{d}{2}=d-1$ elements.

Once all the vertices satisfying case 1 and case 2 have been
(recursively) assigned to their corresponding subsets and removed from
$X$, the remaining subset is $X'$ and the corresponding subgraph is
$G_{X'}$. Now $G_{X'}$ can be viewed as a simple graph of
degree~$d-1$.
Since the chromatic number of any simple graph satisfies $\chi(G')
\leq deg(G')+1$, there exists a partition of ${X'}$ into at most $d$
independent sub-sets $(Z'_a)_{a \in[d]}$. This partition can
be computed with $O(n^2)$ operations by DSATUR and the sub-graphs
$G_{Z'_a}=(Z'_a,E_X)$ have no edges.
Finally, we have that for all $a\in[d]$, $Z_a=Z_a \cup Z'_a$
satisfies:
$X= \sqcup_{a \in[d]} Z_a$ and the subgraphs $G_{Z_a}$ are a-exclusive.
\end{proof}

More generally, one can build~\cref{algo:Str}, that tries to make at
most $d$ strips. But charecterizing precisely the graphs where it
fails remains an open problem.

  \begin{algorithm}[htbp]
        \caption{\textbf{Strips}$( \textbf{X}, \delta)$}
        \label{algo:Str}
            \begin{flushleft}
            \textbf{Input:} $\textbf{X}=\{x_1, \ldots,x_n\}$ and a threshold $\delta$.\\
            \textbf{Output:} Failed or a partition $X= \sqcup_{a \in [d]} Z_a$,  $G_{Z_a}$ are $a$-exclusive.
            \end{flushleft}
        \begin{algorithmic}[1]
  \algnotext{EndFor}
  \algnotext{EndWhile}
\ForAll{$a \in[d]$}
$\mathcal{R}$: $Z_a= \emptyset$;
\EndFor
\State $\mathcal{R}$:  $chgt=True$;
\While{$chgt$}
\State $\mathcal{R}$: $chgt=False$;
\ForAll{$x_i \in X$} $\mathcal{R}$: $a =1$
\While{$a \leq d \wedge\neg  chgt$}
\If{$G_i$ is $a$-exclusive}
\State $\mathcal{R}$: $Z_a=Z_a,x_i$; $X=X {\setminus} \{x_i\}$; $chgt=True$;
\EndIf
\State $\mathcal{R}$: $a=a+1$;
\EndWhile
\EndFor
\EndWhile
\ForAll{$a \in[d]$}
\State $\mathcal{R}$: $(Z_i^a)_{i \in[\chi_a]}=DSATUR(X,\delta,a)$;
\State $\mathcal{R}$: 
Choose $i_a \in [\chi_a]$, s.t. the elements of $Z_{i_a}$ corresponds to vertices with the largest degree of multiplicity;
\State $\mathcal{R}$: $X= X {\setminus} Z_{i_a}^a$
$\mathcal{R}$: $Z_a=Z_a \cup Z_{i_a}^a$
\EndFor
\If{$X=\emptyset$} $\mathcal{R}$: Return  $( Z_a)_{a\in[d]}$.
\Else{} Return Failed

\EndIf
\end{algorithmic}
    \end{algorithm}

\end{document}